\documentclass[abstract,twocolumn,10pt]{scrartcl}
\pdfoutput=1
\usepackage[T1]{fontenc}
\usepackage[utf8]{inputenc}
\usepackage[english]{babel}
\usepackage{mathtools}
\usepackage{amsmath,amssymb,amsthm}
\usepackage[numbers,sort&compress]{natbib}
\usepackage{txfonts}
\usepackage{microtype}
\usepackage[textsize=tiny]{todonotes}
\usepackage{paralist}
\usepackage[ruled,vlined,linesnumbered,algosection]{algorithm2e}
\usepackage[colorlinks, allcolors=blue]{hyperref}
\usepackage{subcaption}
\usepackage[capitalize]{cleveref}
\usepackage{authblk}
\usepackage{flushend}
\usepackage{booktabs}
\usepackage[normalem]{ulem}

\KOMAoptions{DIV=calc}
\addtokomafont{title}{\raggedright}
\setkomafont{paragraph}{\normalfont\itshape}
\setkomafont{subsection}{\normalfont\itshape}
\setkomafont{subsubsection}{\normalfont\itshape}
\setcapindent{0em}

\usetikzlibrary{arrows,shapes,positioning,calc}
\usetikzlibrary{decorations.text,fadings,decorations.markings}
\usetikzlibrary{decorations.pathreplacing}
\newcommand{\apxn}{ap\-prox\-i\-ma\-tion}

\newcommand{\capact}{\ensuremath{Q}}
\newcommand{\bestsol}[1]{\textbf{#1}}
\newcommand{\secsol}[1]{\underline{#1}}

\newtheorem{theorem}{Theorem}
\theoremstyle{definition}
\newtheorem{lemma}[theorem]{Lemma}
\newtheorem{corollary}[theorem]{Corollary}
\newtheorem{definition}[theorem]{Definition}
\newtheorem{observation}[theorem]{Observation}
\newtheorem{proposition}[theorem]{Proposition}
\newtheorem{bfproblem}[theorem]{Problem}

\theoremstyle{remark}
\newtheorem{remark}[theorem]{Remark}
\newtheorem*{bfproblem*}{Problem}

\numberwithin{theorem}{section}
\numberwithin{table}{section}
\numberwithin{figure}{section}
\newcommand{\admsplit}{feasible splitting}
\newcommand{\reprarc}{pivot}
\newcommand{\jia}{q(i)}
\newcommand{\jib}{p(i)}
\crefname{subsection}{Section}{Sections}
\crefname{algorithm}{Algorithm}{Algorithms}
\crefname{observation}{Observation}{Observations}
\crefname{figure}{Figure}{Figures}
\crefname{line}{line}{lines}
\crefname{equation}{Equation}{Equations}

\crefname{enumi}{Condition}{Conditions}
\creflabelformat{enumi}{(#2#1#3)}

\hypersetup{
  pdfauthor={René van Bevern, Christian Komusiewicz, Manuel Sorge},
  pdftitle={A parameterized approximation algorithm for the
mixed and windy Capacitated Arc Routing Problem:
theory and experiments},
  pdfkeywords={vehicle routing, transportation, Rural Postman, Chinese Postman, NP-hard problem, parameterized algorithm, combinatorial optimization},
  pdflang={en}}

\DeclareMathOperator{\balance}{balance}
\DeclareMathOperator{\indeg}{indeg}
\DeclareMathOperator{\outdeg}{outdeg}
\newcommand{\anf}{A}
\newcommand{\nde}{Z}
\newcommand{\Nz}{\mathbb N\cup\{0\}}
\newcommand{\Nzinf}{\mathbb N\cup\{0,\infty\}}
\newcommand{\dist}{\text{dist}}
\newcommand{\cost}{c} %
\newcommand{\comps}{C} %

\newcommand{\longMCF}{{Uncapacitated Minimum-Cost Flow}}
\newcommand{\MCF}{{UMCF}}
\newcommand{\W}{\ensuremath{W}} %
\newcommand{\Wpair}{\W^2}

\newcommand{\pt}{pol\-y\-no\-mi\-al-time}
\newcommand{\prob}[5]{%
  \begin{bfproblem}[#1]\leavevmode
  \par\noindent\hangindent=\parindent\textit{#2}  #3
  \par\noindent\hangindent=\parindent\textit{#4}  #5    
  \end{bfproblem}
}

\newcommand{\optprob}[3]{\prob{#1}{Instance:}{#2}{Task:}{#3}}

\newcommand{\TATSP}{$\triangle$-ATSP}
\newcommand{\longCARP}{Capacitated Arc Routing}
\newcommand{\CARP}{CARP}

\newcommand{\DRPP}{{DRPP}}

\newcommand{\MWCARP}{{MWCARP}}

\newcommand{\MWRPP}{{MWRPP}}

\author{René van Bevern%
\thanks{Supported by
the Russian Foundation for Basic Research (RFBR), project~16-31-60007 mol\textunderscore{}a\textunderscore{}dk, and by the Ministry of Education and Science of the Russian Federation.}
\\
Novosibirsk State University, Novosibirsk, Russian Federation, \texttt{rvb@nsu.ru}\\
Sobolev Institute of Mathematics,
Siberian Branch of the Russian Academy of Sciences,
Novosibirsk, Russian Federation
\and
Christian Komusiewicz%
\thanks{Supported by the German Research Foundation (DFG), project MAGZ (KO~3669/4-1).}
\\
Institut für Informatik, Friedrich-Schiller-Universität Jena, Germany, \texttt{christian.komusiewicz@uni-jena.de}
\and
Manuel Sorge%
\thanks{Supported by the German Research Foundation (DFG), project DAPA (NI~369/12).}
\\
Institut für Softwaretechnik und Theoretische Informatik, TU
Berlin, Germany, \texttt{manuel.sorge@tu-berlin.de}
}

\title{\LARGE A parameterized approximation algorithm for the
mixed and windy Capacitated Arc Routing Problem:
theory and experiments%
\thanks{A preliminary version of this article
    appeared in the Proceedings of the 15th Workshop on
    Algorithmic Approaches for Transportation Modeling,
    Optimization, and Systems (ATMOS'15) \citep{BKS15}.
    This version describes several algorithmic enhancements,
    contains an experimental evaluation of our algorithm,
    and provides a new benchmark data set.}}

\date{}
\begin{document}
\maketitle
\begin{bfseries}
  \begin{sffamily}
    \begin{small}
\boldmath
\noindent
We prove that any polynomial-time $\alpha(n)$-\apxn{} algorithm
for the $n$-vertex metric asymmetric Traveling Salesperson Problem
yields a polynomial-time $O(\alpha(\comps))$-\apxn{} algorithm
for the mixed and windy Capacitated Arc Routing Problem,
where $\comps$~is the number of weakly connected components
in the subgraph induced by the positive-demand arcs---%
a~small number in many applications.
In conjunction with known results, we obtain
constant-factor \apxn{}s for~$\comps\in O(\log n)$ 
and $O\bigl({\log \comps}/{\log\log \comps}\bigr)$-\apxn{}s
in general.
Experiments show that our algorithm,
together with several heuristic enhancements,
outperforms many previous polynomial-time heuristics.
Finally, since the solution quality achievable in polynomial time
appears to mainly depend on~\(C\)
and since \(C=1\) in almost all benchmark instances,
we propose the \texttt{Ob} benchmark set,
simulating cities that are divided into several components by a river.\par
\end{small}
\end{sffamily}
\end{bfseries}


\section{Introduction}
\citet{GW81} introduced the \longCARP{} Problem (\CARP{})
in order to model the search for minimum-cost routes
for vehicles of equal capacity
that are initially located in a vehicle depot
and have to serve all ``customer'' demands.
Applications of \CARP{} include
snow plowing, waste collection, meter reading,
and newspaper delivery \citep{CL14}.
Herein, the customer demands require
that roads of a road network are served.
The road network is modeled as a graph
whose edges represent roads
and whose vertices can be thought of as road intersections.
The customer demands are modeled as positive integers
assigned to edges.
Moreover, each edge has a cost for traveling along it.
\optprob{\longCARP{} {Problem} (\CARP{})}%
{An undirected graph~$G=(V,E)$, a \emph{depot} vertex~$v_0\in V$, travel costs~$\cost\colon E\to \Nz$, edge demands~$d\colon E\to\Nz$, and a vehicle capacity~$\capact{}$.}
{Find a set~$\W{}$ of closed walks in~$G$, each corresponding to the route of one vehicle and passing through the depot vertex~$v_0$, and find a serving function~$s\colon \W{}\to 2^E$ determining for each closed walk~$w\in\W{}$ the subset~$s(w)$ of edges \emph{served} by~$w$ such that
  \begin{compactitem}
  \item $\sum_{w\in\W{}}\cost(w)$ is minimized, where $\cost(w):=\sum_{i=1}^\ell \cost(e_i)$ for a walk~$w=(e_1,e_2,\dots,e_\ell)\in E^\ell$,
  \item $\sum_{e\in s(w)}d(e)\leq \capact{}$, and
  \item each edge~$e$ with $d(e)>0$ is served by exactly one walk in~$\W{}$.
  \end{compactitem}
}
\noindent
Note that vehicle routes may traverse
each vertex or edge of the input graph
multiple times.
Well-known special cases of \CARP{}
are the NP-hard Rural Postman Problem (RPP) \citep{LR76},
where the vehicle capacity is unbounded and, hence,
the goal is to find a shortest possible route
for one vehicle that visits all positive-demand edges,
and the \pt{} solvable Chinese Postman Problem (CPP)
\citep{Edm65,EJ73},
where additionally \emph{all} edges have positive demand.

\subsection{Mixed and windy variants}
\CARP{} is polynomial-time constant-factor approximable
\citep{BHNS14,Jan93,Woe08}.
However, as noted by \citet[Challenge~5]{BNSW14} in a recent survey
on the computational complexity
of arc routing problems,
the polynomial-time approximability of \CARP{}
in directed, mixed, and windy graphs is open.
Herein, a \emph{mixed} graph may contain directed arcs
in addition to undirected edges
for the purpose of modeling one-way roads or
the requirement of servicing a road
in a \emph{specific} direction or
in \emph{both} directions.
In a \emph{windy} graph,
the cost for traversing an undirected edge~$\{u,v\}$
in the direction from~$u$ to~$v$
may be different from the cost
for traversing it in the opposite direction
(this models sloped roads, for example).
In this work, we study \apxn{} algorithms
for mixed and windy variants of \CARP{}.
To formally state these problems, we need some terminology
related to mixed graphs.

\begin{definition}[Walks in mixed and windy graphs]\label[definition]{def:walks}
  A \emph{mixed graph} is a triple~$G=(V,E,A)$, where $V$~is a set of \emph{vertices}, $E\subseteq\{\{u,v\}\mid u,v\in V\}$ is a set of \emph{(undirected) edges}, $A\subseteq V\times V$~is a set of \emph{(directed) arcs} (that might contain loops), and no pair of vertices has an arc \emph{and} an edge between them.  The \emph{head} of an arc~$(u,v)\in V\times V$ is~$v$, its \emph{tail} is~$u$.

  A \emph{walk in~$G$} is a sequence~$w=(a_1,a_2,\dots,a_\ell)$
  such that,
  for each~$a_i=(u,v)$, $1\leq i\leq\ell$,
  we have $(u,v)\in A$ or~$\{u,v\}\in E$,
  and such that the tail of~$a_i$ is the head of~$a_{i-1}$
  for $1<i\leq\ell$.
  If $(u, v)$~occurs in~$w$, then we say that
  $w$~\emph{traverses} the arc~$(u,v)\in A$
  or the edge~$\{u,v\}\in E$.
  If the tail of~$a_1$ is the head of~$a_\ell$,
  then we call~$w$ a \emph{closed walk}.

  Denoting by~$\cost\colon V\times V\to\Nzinf$
  the \emph{travel cost} between vertices of~$G$,
  the \emph{cost of a walk~$w=(a_1,\dots,a_\ell)$} is
  $\cost(w):=\sum_{i=1}^\ell \cost(a_i)$.
  The \emph{cost of a set~$\W{}$ of walks} is
  $\cost(\W{}):=\sum_{w\in\W{}}\cost(w)$.
\end{definition}

\noindent
We study \apxn{} algorithms for the following problem.

\optprob{Mixed and windy CARP (\MWCARP{})}%
{A mixed graph~$G=(V,E,A)$, a depot vertex~$v_0\in V$, travel costs~$\cost\colon V\times V\to\Nzinf$,  demands~$d\colon E\cup A\to\Nz$, and a vehicle capacity~$\capact{}$.}%
{Find a minimum-cost set~$\W{}$ of closed walks in~$G$,
  each passing through the depot vertex~$v_0$,
  and a serving function~$s\colon \W{}\to 2^{E\cup A}$ determining
  for each walk~$w\in\W{}$ the subset~$s(w)$ of the edges and arcs it \emph{serves} such that
  \begin{compactitem}
  \item $\sum_{e\in s(w)}d(e)\leq \capact{}$, and
  \item each edge or arc~$e$ with $d(e)>0$ is served
    by exactly one walk in~$\W{}$.
  \end{compactitem}
}
\noindent
For brevity, we use the term ``arc'' to refer
to both undirected edges and directed arcs.
Besides studying the approximability of \MWCARP{},
we also consider the following special cases.

If the vehicle capacity~\(Q\) in MWCARP is unlimited
(that is, larger than the sum of all demands)
and the depot~\(v_0\) is incident to a positive-demand arc, 
then one obtains the
mixed and windy Rural Postman Problem (MWRPP):

\optprob{Mixed and windy RPP (\MWRPP)}%
{A mixed graph~$G=(V,E,A)$
  with travel costs~$\cost\colon V\times V\to\Nzinf$
  and a set~$R\subseteq E\cup A$ of \emph{required arcs}.}
{Find a minimum-cost closed walk in~$G$ traversing all arcs in~$R$.}
\noindent
If, furthermore, $E=\emptyset$ in MWRPP,
then we obtain the directed Rural Postman Problem (\DRPP) and
if $R=E\cup A$, then we obtain the mixed Chinese Postman Problem (MCPP).

\subsection{An obstacle: approximating metric asymmetric TSP}
\label{sec:reltotsp}
Aiming for good approximate solutions for \MWCARP{},
we have to be aware of the strong relation
of its special case~\DRPP{}
to the following variant of
the Traveling Salesperson Problem (TSP):

\pagebreak[3]
\optprob{Metric asymmetric TSP (\TATSP{})}%
{A set $V$~of vertices and travel costs~$\cost \colon V\times V\to\Nz$ satisfying the triangle inequality $\cost(u,v)\leq \cost(u,w)+\cost(w,v)$ for all $u,v,w\in V$.}%
{Find a minimum-cost cycle that visits every vertex in~$V$ exactly once.}

\noindent
Already \citet{CCCM86} observed that
DRPP is a generalization of \TATSP{}.
In fact, DRPP is at least as hard to approximate as \TATSP{}:
Given a \TATSP{} instance,
one obtains an equivalent \DRPP{} instance
by simply adding a zero-cost loop to each vertex and
by adding these loops to the set~$R$ of required arcs.
This leads to the following observation.

\begin{observation}\label[observation]{obs:conn-to-tatsp}
  Any $\alpha(n)$-\apxn{} for \(n\)-vertex \DRPP{}
  yields an $\alpha(n)$-\apxn{} for \(n\)-vertex \TATSP{}.
\end{observation}

\noindent
Interestingly, the constant-factor approximability of \TATSP{}
is a long-standing open problem and the
$O(\log n/\log \log n)$-\apxn{} by \citet{AGMOS10}
from~\citeyear{AGMOS10}
is the first asymptotic improvement over
the $O(\log n)$-\apxn{} by \citet{FGM82} from~\citeyear{FGM82}.
Thus, the constant-factor approximations
for (undirected) CARP \citep{BHNS14,Jan93,Woe08}
and MCPP \citep{RV99}
cannot be simply carried over
to MWRPP or MWCARP.
\subsection{Our contributions}
As discussed in \cref{sec:reltotsp},
any $\alpha(n)$-\apxn{} for \(n\)-vertex \DRPP{}
yields an $\alpha(n)$-\apxn{} for \(n\)-vertex \TATSP{}.
We first contribute the following theorem
for the converse direction.
\begin{theorem}\label{ourthm}
  If $n$-vertex \TATSP{} is $\alpha(n)$-approximable in $t(n)$~time,
  then
  \begin{enumerate}[(i)]
  \item\label{ourthm1} $n$-vertex \DRPP{} is
    $(\alpha(\comps)+1)$-approximable
    in $t(\comps)+O(n^3\log n)$~time,
  \item\label{ourthm2} $n$-vertex \MWRPP{} is
    $(\alpha(\comps)+3)$-approximable
    in $t(\comps)+O(n^3\log n)$~time,
    and
  \item\label{ourthm3} $n$-vertex \MWCARP{} is
    $(8\alpha(C+1)+27)$-approximable in
    $t(\comps+1)+O(n^3\log n)$~time,
  \end{enumerate}
  \noindent
  where $\comps$~is the number of weakly connected components
  in the subgraph induced by the positive-demand arcs and edges.
\end{theorem}

\noindent
The approximation factors in \cref{ourthm}\eqref{ourthm3}
and \cref{cor:fpt-apx} below are rather large. 
Yet in the experiments described in \cref{sec:experiments},
the relative error of the algorithm was always below~5/4.

We prove \cref{ourthm}(\ref{ourthm1}--\ref{ourthm2}) in \cref{sec:rpp} and \cref{ourthm}\eqref{ourthm3} in \cref{sec:carp}.
Given \cref{ourthm} and \cref{obs:conn-to-tatsp},
the solution quality achievable in polynomial time appears to 
mainly depend on the number~$\comps$. %
The number~$\comps$~is small in several applications, for example, when routing street sweepers and snow plows.
Indeed, we found \(C=1\)
in all but one instance of the benchmark sets
\texttt{mval} and \texttt{lpr} of \citet{BBLP06}
and \texttt{egl-large} of \citet{BE08}.
This makes
the following corollary particularly interesting.

\begin{corollary}\label[corollary]{cor:fpt-apx}
  \MWCARP{} is 35-approximable in $O(2^{\comps}\comps^2+n^3\log n)$~time, that is, constant-factor approximable 
  in polynomial time for $\comps\in O(\log n)$.
\end{corollary}

\noindent
\cref{cor:fpt-apx} 
follows from \cref{ourthm} and the exact $O(2^nn^2)$-time algorithm
for $n$-vertex \TATSP{} by \citet{Bell62} and \citet{HelK62}.
It is ``tight'' in the sense that
finding polynomial-time constant-factor \apxn{}s
for \MWCARP{} in general
would, via \cref{obs:conn-to-tatsp}, answer a question
open since \citeyear{FGM82} and that
computing \emph{optimal} solutions of \MWCARP{}
is NP-hard even if~$C=1$ \citep{BNSW14}.

In \cref{sec:experiments}, 
we evaluate our algorithm on
the \texttt{mval}, \texttt{lpr}, and \texttt{egl-large} benchmark sets
and find that
it outperforms many previous polynomial-time heuristics.
Some instances are solved to optimality.
Moreover, since we found that the solution quality
achievable in polynomial time
appears to crucially depend on the parameter~\(C\)
and almost all of the above benchmark instances have~\(C=1\),
we propose a method
for generating benchmark instances
that simulate cities
separated into few components
by a river, resulting in the \texttt{Ob} benchmark set.

\subsection{Related work}
Several polynomial-time heuristics for variants of CARP
are known \citep{GW81,MA05,BBLP06,BE08} and, in particular, 
used for computing initial solutions
for more time-consuming local search
and genetic algorithms \citep{BBLP06,BE08}. 
Most heuristics are improved variants of three basic approaches:
\begin{description}
\item[Augment and merge] heuristics start out with small vehicle tours,
each serving one positive-demand arc,
then successively grow and merge these tours
while maintaining capacity constraints \citep{GW81}.
\item[Path scanning] heuristics grow vehicle tours
by successively augmenting them with
the ``most promising'' positive-demand arc \citep{GDB83},
for example,
by the arc that is closest to the previously added arc.
\item[Route first, cluster second] approaches
first construct a \emph{giant tour}
that visits all positive-demand arcs,
which can then be split \emph{optimally} into subsegments
satisfying capacity constraints \citep{Bea83,Ulu85}.
\end{description}
The giant tour for the ``route first, cluster second'' approach
can be computed heuristically \citep{BBLP06,BE08},
yet when computing it
using a constant-factor approximation for the undirected RPP,
one can split it to obtain a constant-factor approximation
for the undirected CARP \citep{Jan93,Woe08}.
Notably, the ``route first, cluster second'' approach
is the only one known to yield
solutions of guaranteed quality for CARP in polynomial time.
One barrier for generalizing this result to MWCARP is that
already approximating MWRPP is challenging (see \cref{sec:reltotsp}).
Indeed, the only polynomial-time algorithms
with guaranteed solution quality
for arc routing problems in mixed graphs
are for variants to which \cref{obs:conn-to-tatsp} does not apply
since \emph{all} arcs and edges have to be served
\citep{RV99,DLL14}.

Our algorithm follows the ``route first, cluster second'' approach:
We first compute an approximate giant tour
using \cref{ourthm}\eqref{ourthm2} and
then, analogously to the approximation algorithms
for undirected CARP \citep{Jan93,Woe08},
split it to obtain \cref{ourthm}\eqref{ourthm3}.
However, since the analyses
of the approximation factor for undirected CARP
rely on symmetric distances between vertices \citep{Jan93,Woe08},
our analysis is fundamentally different.
Our experiments show that
computing the giant tour using \cref{ourthm}\eqref{ourthm2}
is beneficial
compared to computing it heuristically
like \citet{BBLP06} and \citet{BE08}.

Notably,
the approximation factor of \cref{ourthm} depends
on the number~\(C\) of connected components
in the graph
induced by positive-demand arcs.
This number~\(C\) is small in many applications and benchmark data sets, a fact that inspired the development of exact exponential-time algorithms for RPP which are efficient when~\(C\) is small
\citep{Fre77, GWY16, SBNW11,SBNW12}.
\citet{Orl76} noticed already in \citeyear{Orl76}
that the number~\(C\) is a determining factor
for the computational complexity of RPP.
\cref{ourthm} shows that it is also a determining factor
for the solution quality achievable in polynomial time.

In terms of parameterized complexity theory \citep{DF13,CFK+15},
one can interpret \cref{cor:fpt-apx} as
a fixed-parameter constant-factor approximation algorithm \citep{Mar08}
for MWCARP parameterized by~\(C\).

\section{Preliminaries}\label{sec:preliminaries}
Although we consider problems on mixed graphs
as defined in \cref{def:walks},
in some of our proofs
we use more general mixed \emph{multigraphs}~$G=(V,E,A)$
with a set~$V =: V(G)$ of \emph{vertices},
a multiset~$E =: E(G)$ over $\{\{u,v\}\mid u,v\in V\}$
of \emph{(undirected) edges},
a multiset~$A =: A(G)$ over $V\times V$
of \emph{(directed) arcs} that may contain self-loops,
and \emph{travel costs}~$\cost\colon V\times V\to\Nzinf$.
If $E=\emptyset$, then $G$~is a \emph{directed multigraph}.

From \cref{def:walks}, recall the definition of walks in mixed graphs.
An \emph{Euler tour for~$G$} is a closed walk that traverses
each arc and each edge of~$G$
exactly as often as it is present in~$G$.
A graph is \emph{Eulerian} if it allows for an Euler tour.  Let $w=(a_1,a_2,\dots,a_\ell)$ be a walk.  The \emph{starting point of~$w$} is the tail of~$a_1$, the \emph{end point of~$w$} is the head of~$a_\ell$.  A \emph{segment} of~$w$ is a consecutive subsequence of~$w$. Two segments~$w_1=(a_i,\dots,a_j)$ and $w_2=(a_{i'},\dots,a_{j'})$ of the walk~$w$ are \emph{non-overlapping} if $j<i'$ or $j'<i$.  Note that two segments of~$w$ might be non-overlapping yet share arcs if $w$~contains an arc several times.  The \emph{distance}~$\dist_G(u,v)$ from vertex~$u$ to vertex~$v$ of~$G$ is the minimum cost of a walk in~$G$ starting in~$u$ and ending in~$v$.

The \emph{underlying undirected (multi)graph} of~$G$ is obtained by replacing all directed arcs by undirected edges.  Two vertices~$u,v$ of~$G$ are \emph{(weakly) connected} if there is a walk starting in~$u$ and ending in~$v$ in the underlying undirected graph of~$G$.  A \emph{(weakly) connected component} of~$G$ is a maximal subgraph of~$G$ in which all vertices are mutually (weakly) connected.

For a multiset~$R\subseteq V\times V$ of arcs, $G[R]$ is the directed multigraph consisting of the arcs in~$R$ and their incident vertices of~$G$.  We say that $G[R]$~is the graph \emph{induced by the arcs in~$R$}. For a walk~$w=(a_1,\dots,a_\ell)$ in~$G$, $G[w]$ is the directed multigraph consisting of the arcs~$a_1,\dots,a_\ell$ and their incident vertices, where $G[w]$~contains each arc with the multiplicity it occurs in~$w$.  Note that $G[R]$ and~$G[w]$ might contain arcs with a higher multiplicity than~$G$ and, therefore, are not necessarily sub(multi)graphs of~$G$.  Finally, the cost of a multiset~$R$ is $\cost(R):=\sum_{a\in R}\mathbb \nu(a)\cost(a)$, where $\mathbb \nu(a)$~is the multiplicity of~$a$ in~$R$.

\section{Rural Postman}
\label{sec:rpp}

This section presents our \apxn{} algorithms for \DRPP{} and \MWRPP{},
thus proving \cref{ourthm}\eqref{ourthm1} and \eqref{ourthm2}.
\cref{sec:eulerian-comps} shows an algorithm
for the special case of \DRPP{}
where the required arcs induce a subgraph
with Eulerian connected components.
\cref{sec:drpp,sec:mwrpp} subsequently generalize this algorithm
to \DRPP{} and \MWRPP{}
by adding to the set of required arcs an arc set of minimum weight
so that the required arcs induce a graph
with Eulerian connected components.

\subsection{Special case: Required arcs induce Eulerian components}
\label{sec:eulerian-comps}
To turn $\alpha(n)$-\apxn{}s for $n$-vertex \TATSP{}
into $(\alpha(\comps)+1)$-\apxn{}s for this special case of \DRPP{},
we use \cref{alg:eulerian-rp}. The two main steps of the algorithm are illustrated in~\Cref{fig:eulerian-rp}:
The algorithm first computes an Euler tour
for each connected component of the graph~\(G[R]\)
induced by the set~\(R\) of required arcs and
then connects them using an approximate \TATSP{} tour
on a vertex set~\(V_R\)
containing (at least) one vertex
of each connected component of~\(G[R]\).

\begin{figure*}[tb]
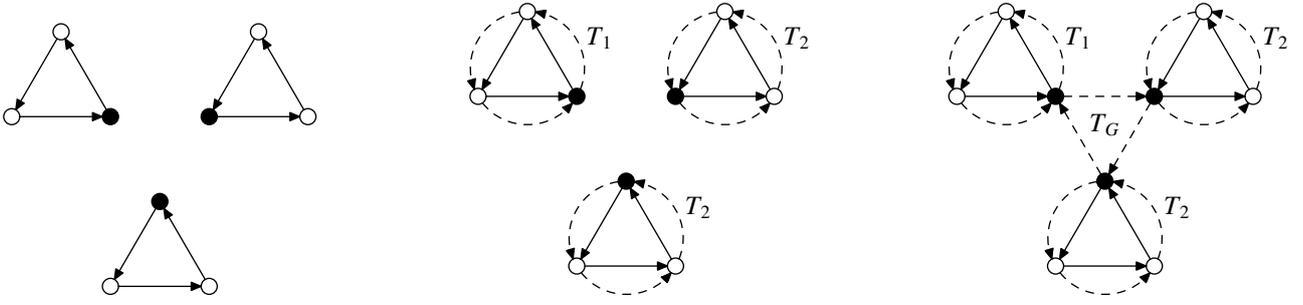

  \centering
  \begin{subfigure}[b]{0.3\textwidth}
    \centering
    \includegraphics{dar-1}
    \caption{Input: Only required arcs~$R$ are shown, vertices in~$V_R$ are black.}
  \end{subfigure}
  \hfill
  \begin{subfigure}[b]{0.3\textwidth}
    \centering
    \includegraphics{dar-2}
    \caption{Compute Euler tours~$T_i$  (dashed)
      for each connected component of~{$G[R]$}.}
  \end{subfigure}
  \hfill
  \begin{subfigure}[b]{0.3\textwidth}
    \centering
    \includegraphics{dar-3}
    \caption{Add closed walk~$T_G$ with the vertices in~\(V_R\)
      to get a feasible solution~$T$ (dashed).}
  \end{subfigure}
  \caption{Steps of \cref{alg:eulerian-rp}
    to compute feasible solutions for \DRPP{}
    when all connected components of~$G[R]$ are Eulerian.}
  \label{fig:eulerian-rp}
\end{figure*}

The following \cref{lem:eulerian-rp} gives a bound
on the cost of the solution returned by \cref{alg:eulerian-rp}.
\cref{alg:eulerian-rp} and \cref{lem:eulerian-rp}
are more general than necessary for this special case of DRPP.
In particular, we will not exploit yet
that they allow $R$~to be a \emph{multi}set and
$V_R$~to contain more than one vertex
of each connected component of~$G[R]$.
This will become relevant in \cref{sec:drpp},
when we use \cref{alg:eulerian-rp}
as a subprocedure to solve the general \DRPP{}.
\begin{algorithm*}
  \KwIn{A directed graph~$G$ with travel costs~$\cost{}$,
    a multiset~$R$ of arcs of~$G$ such that
    $G[R]$~consists of~$\comps$ Eulerian connected components,
    and a set~$V_R\subseteq V(G[R])$ containing
    at least one vertex of each connected component of~$G[R]$.}

  \KwOut{A closed walk traversing all arcs in~$R$.}
  
  \For{$i=1,\dots,\comps$}{
    $v_i\gets{}$any vertex of~$V_R$ in component~$i$ of~$G[R]$\;
    $T_i\gets{}$Euler tour of connected component~$i$ of~$G[R]$ starting and ending in~$v_i$\nllabel{a1:comptours}\;} 

  $(V_R,\cost')\gets{}$\TATSP{} instance on the vertices~$V_R$, where $\cost'(v_i,v_j):=\dist_G(v_i,v_j)$\nllabel{a1:buildtsp}\;

  $T_{V_R}\gets\alpha(|V_R|)$-approximate \TATSP{} solution for~$(V_R,\cost')$\nllabel{a1:joincomps}\;

  $T_G\gets{}$closed walk for~$G$ obtained by replacing each arc~$(v_i,v_j)$ on~$T_{V_R}$ by a shortest path from \(v_i\) to \(v_j\) in~$G$\nllabel{a1:tg}\;

  $T\gets{}$closed walk obtained by following~$T_G$ and taking a detour~$T_i$ whenever reaching a vertex~$v_i$\nllabel{a1:buildtour}\;

  \Return{$T$}\nllabel{a1:rettour}\;

  \caption{Algorithm for the proof of \cref{lem:eulerian-rp}}
  \label[algorithm]{alg:eulerian-rp}
\end{algorithm*}

\begin{lemma}\label[lemma]{lem:eulerian-rp}
  Let $G$~be a directed graph with travel costs~$\cost{}$, let $R$~be a multiset of arcs of~$G$ such that $G[R]$~consists of~$\comps$ Eulerian connected components, let $V_R\subseteq V(G[R])$ be a vertex set containing at least one vertex of each connected component of~$G[R]$, and let~$\tilde T$ be any closed walk containing the vertices~$V_R$.

  If $n$-vertex \TATSP{} is $\alpha(n)$-approximable in $t(n)$~time, then \cref{alg:eulerian-rp} applied to~$(G,c,R)$ and~$V_R$ returns a closed walk of cost at most $\cost(R)+\alpha(|V_R|)\cdot \cost(\tilde T)$ in $t(|V_R|)+O(n^3)$~time that traverses all arcs of~$R$.
\end{lemma}

\begin{proof}
  We first show that the closed walk~$T$ returned by \cref{alg:eulerian-rp} visits all arcs in~$R$. Since the \TATSP{} solution~$T_{V_R}$ constructed in \cref{a1:joincomps} visits all vertices~$V_R$, in particular~$v_1,\dots,v_\comps$, so does the closed walk~$T_G$ constructed in \cref{a1:tg}.  Thus, for each vertex~$v_i$, $1\leq i\leq\comps$, $T$~takes Euler tour~$T_i$ through the connected component~$i$ of~$G[R]$ and, thus, visits all arcs in~$R$.
  
 We analyze the cost~$\cost(T)$.  The closed walk~$T$ is composed of the Euler tours~$T_i$ computed in \cref{a1:comptours} and the closed walk~$T_G$ computed in \cref{a1:tg}.  Hence, $\cost(T)=\cost(T_G)+\sum_{i=1}^{\comps}\cost(T_i)$.
Since each~$T_i$ is an Euler tour for some connected component~$i$ of~$G[R]$, each~$T_i$ visits each arc of component~$i$ as often as it is contained in~$R$.  Consequently, $\sum_{i=1}^{\comps}\cost(T_i)=\cost(R)$.

It remains to analyze $\cost(T_G)$.
Observe first that
the distances in the \TATSP{} instance~$(V_R,\cost')$
correspond to shortest paths in~$G$
and thus
fulfill the triangle inequality.
We have $\cost(T_G)=\cost'(T_{V_R})$ by
construction of the \TATSP{} instance~$(V_R,\cost')$ in
\cref{a1:buildtsp} and by construction of~$T_G$ from~$T_{V_R}$ in
\cref{a1:tg}. Let $\tilde T$~be any closed walk containing~$V_R$ and let
$T^*_{V_R}$ be an optimal solution for the \TATSP{}
instance~$(V_R,\cost')$.  If we consider the closed
walk~$\tilde T_{V_R}$ that visits the vertices~$V_R$ of the \TATSP{}
instance~$(V_R,\cost')$ in the same order as~$\tilde T$, we get
$\cost'(T^*_{V_R})\leq\cost'(\tilde T_{V_R})\leq \cost(\tilde T)$.
Since the closed walk~$T_{V_R}$ computed in \cref{a1:joincomps} is an
$\alpha(|V_R|)$-approximate solution to the \TATSP{}
instance~$(V_R,\cost')$, it finally follows that
$\cost(T_G)=\cost'(T_{V_R})\leq\alpha(|V_R|)\cdot
\cost'(T^*_{V_R})\leq\alpha(|V_R|)\cdot \cost(\tilde T)$.
  
Regarding the running time, observe that the instance~$(V_R,\cost')$ in \cref{a1:buildtsp} can be constructed in $O(n^3)$~time using the Floyd-Warshall all-pair shortest path algorithm \citep{Flo62}, which dominates all other steps of the algorithm except for, possibly, \cref{a1:joincomps}.
\end{proof}

\noindent
\cref{lem:eulerian-rp} proves \cref{ourthm}\eqref{ourthm1}
for \DRPP{}~instances $I=(G,c,R)$ when $G[R]$~consists of Eulerian connected components: Pick~$V_R$ to contain exactly one vertex
of each of the $\comps$~connected components of~$G[R]$.
Since an optimal solution~$T^*$ for~$I$
visits the vertices~$V_R$ and satisfies $\cost(R)\leq \cost(T^*)$,
\cref{alg:eulerian-rp} yields a solution
of cost at most $\cost(T^*)+\alpha(C)\cdot \cost(T^*)$.

\subsection{Directed Rural Postman}\label{sec:drpp}

In the previous section,
we proved \cref{ourthm}\eqref{ourthm1} for 
the special case of \DRPP{} when
$G[R]$~consists of Eulerian connected components.
We now transfer this result to the general DRPP.
To this end,
observe that a feasible solution~$T$ for a \DRPP{} instance~$(G,c,R)$
enters each vertex~$v$ of~$G$ as often as it leaves.  Thus, if we
consider the multigraph~$G[T]$ that contains
each arc of~$G$ with same multiplicity as~$T$, then $G[T]$ is a
supermultigraph of~$G[R]$ in which every vertex is \emph{balanced}
\citep{DMNW13,SBNW12}:
\begin{definition}[Balance]
  We denote the \emph{balance} of a vertex~$v$ in a graph~$G$ as
  \[\balance_G(v):=\indeg_G(v)-\outdeg_G(v).\]
  We call a vertex~$v$ \emph{balanced} if $\balance_G(v)=0$.
\end{definition}

\noindent Since $G[T]$ is a supergraph of~$G[R]$ in which all vertices are balanced and since a directed connected multigraph is Eulerian if and only if all its vertices are balanced, we immediately obtain the below observation.  Herein and in the following, for two (multi-)sets~$X$ and~$Y$, $X \uplus Y$ is the multiset obtained by adding the multiplicities of each element in~$X$ and~$Y$.

\begin{observation}\label[observation]{obs:eulerize}
  Let $T$~be a feasible solution for a \DRPP{} instance~$(G,c,R)$ such that~$G[R]$ has $\comps$~connected components and let $R^*$ be a minimum-cost multiset of arcs of~$G$ such that every vertex in~$G[R\uplus R^*]$ is balanced.
  Then, $\cost(R\uplus R^*)\leq \cost(T)$ and $G[R\uplus R^*]$ consists of at most~$\comps$ Eulerian connected components.
\end{observation}

\begin{algorithm*}
  \KwIn{A \DRPP{} instance~$I=(G,\cost,R)$ such that $G[R]$~has $\comps$~connected components and a set~$V_R$ of vertices, one of each connected component of~$G[R]$.}

  \KwOut{A feasible solution for~$I$.}

  $f\gets{}$minimum-cost flow for the \MCF{} instance~$(G,\balance_{G[R]},c)$\nllabel{a2:computeflow}\;

  \lForEach(\nllabel{a2:buildrprime}){$a\in A(G)$}
  {add arc~$a$ with multiplicity~$f(a)$ to (initially empty) multiset~$R^*$}

  $T\gets{}$closed walk computed by \cref{alg:eulerian-rp} applied to~$(G,c,R\uplus R^*)$ and~$V_R$\nllabel{a2:computetour}\;
  \Return{$T$}\;
  \caption{Algorithm for the proof of \cref{lem:drpp}.}
  \label[algorithm]{alg:drpp}
\end{algorithm*}

\noindent \cref{alg:drpp} computes an $(\alpha(\comps)+1)$-\apxn{} for a
\DRPP{} instance~$(G,c,R)$ by first computing a minimum-cost arc
multiset~$R^*$ such that~$G[R\uplus R^*]$ contains only balanced
vertices and then applying \cref{alg:eulerian-rp}
to~$(G,c,R\uplus R^*)$.
It is well known that the first step can be modeled
using the \longMCF{} Problem \citep{CCCM86,CMR00,DMNW13,EJ73,Fre79}:

\optprob{\longMCF{} (\MCF{})}%
{A directed graph~$G=(V,A)$ with \emph{supply} $s\colon V\to\mathbb Z$ and \emph{costs}~$\cost\colon A\to\Nz$.}
{Find a \emph{flow} $f\colon A\to\Nz$ minimizing $\sum_{a\in A}\cost(a)f(a)$ such that, for each~$v\in V$,
  \begin{align}
    \smashoperator{\sum_{(v,w)\in A}}f(v,w)-\smashoperator{\sum_{(w,v)\in A}}f(w,v)&=s(v).\tag{FC}\label{fc}
  \end{align}
}

\noindent
\cref{fc} is known as the \emph{flow conservation} constraint:
For every vertex~$v$ with~$s(v)=0$,
there are as many units of flow entering the node as leaving it.
Nodes~$v$ with~$s(v)>0$ ``produce'' $s(v)$~units of flow,
whereas nodes~$v$ with~$s(v)<0$ ``consume'' $s(v)$~units of flow.
For our purposes, we will use \(s(v):=\balance_{G[R]}(v)\).
\MCF{} is solvable in $O(n^3\log n)$~time \citep[Theorem~10.34]{AMO93}.

\begin{lemma}\label[lemma]{lem:drpp}
  Let $I:=(G,c,R)$~be a \DRPP{} instance such that $G[R]$~has $\comps$~connected components, and let~$V_R$ be a vertex set containing exactly one vertex of each connected component of~$G[R]$.  Moreover, consider two closed walks in~$G$:
  \begin{itemize}
  \item Let $\tilde T$ 
be any closed walk containing the vertices~$V_R$, and
  \item let $\hat T$~be any feasible solution for~$I$.
  \end{itemize}
  If $n$-vertex \TATSP{} is $\alpha(n)$-approximable in $t(n)$~time, then \cref{alg:drpp} applied to~$I$ and~$V_R$ returns a feasible solution of cost at most $\cost(\hat T)+\alpha(\comps)\cdot \cost(\tilde T)$ in $t(\comps)+O(n^3\log n)$~time.
\end{lemma}

\begin{proof}
  For the sake of self-containment,
  we first prove that \cref{alg:drpp} in \cref{a2:buildrprime}
  indeed computes a minimum-cost arc set~$R^*$
  such that all vertices in~$G[R\uplus R^*]$ are balanced.
  This follows from the one-to-one correspondence
  between arc multisets~$R'$ such that
  $G[R\uplus R']$~has only balanced vertices
  and flows~$f$ for the \MCF{} instance $I':=(G,\balance_{G[R]},c)$:
  Each vertex~$v$ has $\balance_{G[R]}(v)$
  more incident in-arcs than out-arcs in~$G[R]$ and, thus,
  in order for $\balance_{G[R\uplus R']}(v)=0$ to hold,
  $R'$~has to contain $\balance_{G[R]}(v)$ more out-arcs
  than in-arcs incident to~$v$. 
  Likewise, by \eqref{fc}, in any feasible flow for~$I'$,
  there are $\balance_{G[R]}(v)$~more units of flow leaving~$v$
  than entering~$v$.

  Thus,
  from a multiset~$R'$ of arcs such that~$G[R\uplus R']$ is balanced,
  we get a feasible flow~$f$ for~$I'$ by setting~$f(v,w)$
  to the multiplicity of the arc~$(v,w)$ in~$R'$.
  From a feasible flow~$f$ for~$I'$,
  we get a multiset~$R'$ of arcs such that~$G[R\uplus R']$ is balanced
  by adding to~$R'$ each arc~$(v,w)$ with multiplicity~$f(v,w)$.
  We conclude that the arc multiset~$R^*$
  computed in \cref{a2:buildrprime} is a minimum-cost set
  such that~$G[R\uplus R^*]$ is balanced:
  A set of lower cost would yield a flow cheaper
  than the optimum flow~$f$ computed in \cref{a2:computeflow}.

  We use the optimality of~$R^*$ to give an upper bound on the cost of the closed walk~$T$ computed in \cref{a2:computetour}.  Since $V_R$~contains exactly one vertex of each connected component of~$G[R]$, it contains at least one vertex of each connected component of~$G[R\uplus R^*]$.  Therefore, \cref{alg:eulerian-rp} is applicable to~$(G,c,R\uplus R^*)$ and, by \cref{lem:eulerian-rp}, yields a closed walk in~$G$ traversing all arcs in~$R\uplus R^*$ and having cost at most~$\cost(R\uplus R^*)+\alpha(|V_R|)\cdot\cost(\tilde T)$.  This is a feasible solution for~$(G,c,R)$ and, since by \cref{obs:eulerize}, we have $\cost(R\uplus R^*)\leq\cost(\hat T)$, it follows that this feasible solution has cost at most $\cost(\hat T)+\alpha(C)\cdot\cost(\tilde T$).

Finally, the running time of \cref{alg:drpp} follows from the fact that the minimum-cost flow in \cref{a2:computeflow} is computable in $O(n^3\log n)$~time \citep[Theorem~10.34]{AMO93} and that \cref{alg:eulerian-rp} runs in $t(\comps)+O(n^3)$~time (\cref{lem:eulerian-rp}).
\end{proof}

\noindent
We may now  prove \cref{ourthm}\eqref{ourthm1}.

\begin{proof}[Proof of \cref{ourthm}(\ref{ourthm1})]
  Let $(G,c,R)$~be an instance of \DRPP{} and let~$V_R$ be a set of vertices containing exactly one vertex of each connected component of~$G[R]$.  An optimal solution~$T^*$ for~$I$ contains all arcs in~$R$ and all vertices in~$V_R$ and hence, by \cref{lem:drpp}, \cref{alg:drpp} computes a feasible solution~$T$ with $\cost(T)\leq \cost(T^*)+\alpha(\comps)\cdot \cost(T^*)$ for~$I$.
\end{proof}

\noindent
Before generalizing \cref{alg:drpp} to MWRPP,
we point out two design choices in the algorithm
that allowed us to prove an approximation factor.
\cref{alg:drpp} has two steps:
It first adds a minimum-weight set~\(R^*\) of required arcs
so that \(G[R\uplus R^*]\)~has Eulerian connected components.
Then, these connected components
are connected using a cycle via \cref{alg:eulerian-rp}.

In the first step, it might be tempting
to add a minimum-weight set~\(R'\) of required arcs 
so that each connected component of~\(G[R]\)
becomes an Eulerian connected component of~\(G[R\uplus R']\).
However,
this set \(R'\) might be more expensive than~\(R^*\):
Multiple non-Eulerian connected components of~\(G[R]\)
might be contained in one Eulerian connected component
of~\(G[R\uplus R^*]\).

In the second step, it is crucial to connect
the connected components of~\(G[R\uplus R^*]\)
using a \emph{cycle}.
\citet{CCCM86} and \citet{CMR00}, for example,
reverse the two phases of the algorithm
and first join the connected components of~\(G[R]\)
using a minimum-weight arborescence or spanning tree,
respectively.
This, however, may increase the imbalance of vertices
and, thus,
the weight of the arc set~\(R^*\) that has to be added
in their second phase
in order to balance the vertices of~\(G[R\uplus R^*]\).

Interestingly, the heuristic of \citet{CMR00} aims to find
a minimum-weight connecting arc set so that the resulting
graph can be balanced at low extra cost
and already \citet{PW95} pointed out that,
in context of the (undirected) RPP,
reversing the steps
in the algorithm of \citet{CCCM86}
can be beneficial.

\subsection{Mixed and windy Rural Postman}\label{sec:mwrpp}

In the previous section,
we presented \cref{alg:drpp} for \DRPP{}
in order to prove \cref{ourthm}\eqref{ourthm1}.
We now generalize it to MWRPP
in order to prove \cref{ourthm}\eqref{ourthm2}.

To this end,
we replace each undirected edge~$\{u,v\}$
in an \MWRPP{} instance by two directed arcs~$(u,v)$ and~$(v,u)$,
where we force the undirected \emph{required} edges
of the \MWRPP{} instance
to be traversed in the cheaper direction:

\begin{lemma}\label[lemma]{lem:direct-edges}
  Let $I:=(G,c,R)$~be an~\MWRPP{} instance and let $I':=(G',c,R')$~be the \DRPP{} instance obtained from~$I$ as follows: 
  \begin{itemize}
  \item $G'$ is obtained by replacing each edge~$\{u,v\}$ of~$G$ by two arcs~$(u,v)$ and~$(v,u)$,
  \item $R'$ is obtained from~$R$ by replacing each edge~$\{u,v\}\in R$ by an arc~$(u,v)$ if~$\cost(u,v)\leq \cost(v,u)$ and by~$(v,u)$ otherwise.
  \end{itemize}
  Then,
  \begin{enumerate}[(i)]
  \item\label{de1} each feasible solution~\(T'\) for~$I'$
    is a feasible solution of the same cost for~$I$ and,
  \item\label{de2} for each feasible solution~$T$ for~$I$, there
    is a feasible solution~$T'$ for~$I'$ with~$\cost(T')< 3\cost(T)$.
\end{enumerate}
\end{lemma}

\begin{proof}
  Statement \eqref{de1} is obvious
  since each required edge of~$I$ is served by~$T'$
  in at least one direction.
  Moreover, the cost functions in~$I$ and~$I'$ are the same.

  Towards \eqref{de2}, let $T$~be a feasible solution for~$I$,
  that is,
  \(T\)~is a closed walk that traverses all required arcs and edges of~\(I\).
  We show how to transform~\(T\) into a feasible solution for~\(I'\).
  Let \((u,v)\) be an arbitrary required arc of~\(I'\)
  that is not traversed by~\(T\).
  Then, \(I\)~contains a required edge~\(\{u,v\}\)
  and \(T\)~contains arc~\((v,u)\) of~\(I'\).
  Moreover, \(c(u,v)\leq c(v,u)\).
  Thus, we can replace \((v,u)\) on~\(T\) by
  the sequence of arcs $(v,u),(u,v),(v,u)$.
  This sequence serves the required arc~\((u,v)\) of~\(I'\)
  and costs
  $\cost(v,u)+\cost(u,v)+\cost(v,u)\leq 3\cost(u,v)$.
\end{proof}

\noindent Using \cref{lem:direct-edges}, it is easy to prove \cref{ourthm}\eqref{ourthm2}.

\begin{proof}[Proof of \cref{ourthm}(\ref{ourthm2})]
  Given an \MWRPP{} instance~$I=(G,c,R)$, compute a \DRPP{} instance~$I':=(G',c,R')$ as described in \cref{lem:direct-edges}.  This can be done in linear time.

  Let $V_R$~be a set of vertices containing exactly one vertex of each connected component of~$G'[R']$ and let $T^*$ be an optimal solution for~$I$.  Observe that $T^*$~is not necessarily a feasible solution for~$I'$, since it might serve required arcs of~$I'$ in the wrong direction.  Yet $T^*$~is a closed walk in~$G'$ visiting all vertices of~$V_R$.  Moreover, by \cref{lem:direct-edges}, $I'$~has a feasible solution~$T'$ with~$\cost(T')\leq 3\cost(T^*)$.

  Thus, applying \cref{alg:drpp} to~$I'$ and~$V_R$ yields a feasible solution~$T$ of cost at most~$\cost(T')+\alpha(C)\cdot \cost(T^*)\leq 3\cost(T^*)+\alpha(C)\cdot \cost(T^*)$ due to \cref{lem:drpp}. Finally, $T$~is also a feasible solution for~$I$ by \cref{lem:direct-edges}.
\end{proof}

\begin{remark}
\label{rem:orientation}
If a required edge~\(\{u,v\}\) has \(c(u,v)=c(v,u)\),
then we replace it by two arcs \((u,v)\) and \((v,u)\)
in the input graph~\(G\)
and replace~\(\{u,v\}\) by an arbitrary one of them
in the set~\(R\) of required arcs
without influencing the approximation factor.
This gives a lot of room for experimenting with heuristics
that ``optimally'' orient undirected required edges
when converting MWRPP to DRPP \citep{MA05,CMR00}.
Indeed, we will do so in \cref{sec:experiments}.
\end{remark}
\section{Capacitated Arc Routing}
\label{sec:carp}

\noindent
We now present our \apxn{} algorithm for MWCARP, thus proving \cref{ourthm}\eqref{ourthm3}.
Our algorithm follows the ``route first, cluster second''-approach
\citep{Bea83,Ulu85,FHK78,Jan93,Woe08}
and exploits the fact that joining all vehicle tours of a solution gives an \MWRPP{} tour
traversing all positive-demand arcs and the depot.
Thus,
in order to approximate \MWCARP{}, the idea is to
first compute an approximate \MWRPP{} tour
and then split it into subtours,
each of which can be served by a vehicle of capacity~$\capact{}$.
Then we close each subtour by shortest paths via the depot.
We now describe our \apxn{} algorithm for MWCARP in detail. For convenience, we use the following notation.

\begin{definition}[Demand arc]
  For a mixed graph~\(G=(V,A,E)\)
  with demand function~$d \colon E \cup A \to \mathbb{N} \cup \{0\}$,
  we define \[R_d := \{a \in E \cup A \mid d(a) > 0\}\] to be
  the set of \emph{demand arcs}.
\end{definition}

\noindent
We construct \MWCARP{} solutions from what we call \emph{feasible splittings}
of \MWRPP{} tours~\(T\).

\begin{definition}[Feasible splitting]\label[definition]{def:admsplit}
  For an \MWCARP{} instance~$I = (G,v_0,c,d,\capact{})$, let $T$~be a
  closed walk containing all arcs in~$R_d$ and
  $\W{}=(w_1,\ldots,w_\ell)$ be a tuple of segments of~$T$. In the
  following, we abuse notation and refer by~\W{} to both the tuple and
  the set of walks it contains.

  Consider a serving function $s \colon \W{} \to 2^{R_d}$
  that assigns to each walk~\(w\)
  the set~\(s(w)\) of arcs in~$R_d$ that it serves.
  We call $(\W, s)$ a \emph{\admsplit{} of~$T$}
  if the following conditions hold:
\begin{enumerate}[(i)]\label{admsplit}
\item\label{nonoverlap} the walks in~$\W$ are mutually non-overlapping segments of~$T$,
\item when concatenating the walks in~$\W$ in order, we obtain a subsequence of~$T$,\label{as:order}
\item each $w_i \in \W$ begins and ends with an arc in $s(w_i)$,\label{as:short}
\item $\{s(w_i) \mid w_i \in \W\}$ is a partition of~$R_d$, and\label{as:partition}
\item for each $w_i \in \W$, we have $\sum_{e\in s(w_i)}d(e) \leq \capact{}$ and, if %
  $i<\ell$, then $\sum_{e\in s(w_i)}d(e)+d(a) > \capact{}$, where~$a$ is the first arc served by~$w_{i+1}$.\label{as:full}
\end{enumerate}
\end{definition}

\paragraph{Constructing feasible splittings.}
Given an \MWCARP{} instance~$I=(G,v_0,c,d,\capact{})$,
a \admsplit{}~$(\W,s)$ of a closed walk~$T$
that traverses all arcs in~$R_d$
can be computed in linear time
using the following greedy strategy.
We assume that each arc
has demand at most~$\capact{}$
since otherwise $I$~has no feasible solution.
Now, traverse~$T$, successively
defining subwalks~$w \in \W$ and the corresponding sets~$s(w)$
one at a time.
The traversal starts with the first arc~$a \in R_d$ of~$T$ and by
creating a subwalk~$w$ consisting only of~$a$ and $s(w) = \{a\}$. On
discovery of a still unserved arc
$a \in R_d \setminus (\bigcup_{w' \in \W}s(w'))$ do the following. If
$\sum_{e\in s(w)}d(e)+d(a) \leq \capact{}$, then add~$a$ to~$s(w)$ and
append to~$w$ the subwalk of~$T$ that was traversed since discovery of
the previous unserved arc in~$R_d$. Otherwise, mark $w$ and $s(w)$ as
finished, start a new tour~$w \in \W$ with $a$ as the first arc, set
$s(w)=\{a\}$, and continue the traversal of~$T$. If no such arc $a$ is found,
then stop. It is not hard to verify that $(\W, s)$~is indeed a
\admsplit{}.
\begin{algorithm*}
  \KwIn{An \MWCARP{} instance~$I = (G,v_0,c,d,\capact{})$ such that $(v_0,v_0)\in R_d$ and such that $G[R_d]$~has $\comps$~connected components.}  
\KwOut{A feasible solution for $I$.}

\tcc{Compute a base tour containing all demand arcs and the depot}

  $I' \leftarrow{}$\MWRPP{} instance~$(G,c,R_d)$\nllabel{alg:carp:drpinst}\;
  
  $T \leftarrow\beta(\comps)$-approximate \MWRPP{} tour for~$I'$ starting and ending in~$v_0$\nllabel{alg:carp:basetour}\;

  \tcc{Split the base tour into one tour for each vehicle}
  $(\W, s) \leftarrow $ a \admsplit{} of~$T$\nllabel{alg:carp:split}\;

\ForEach{$w \in \W$}{%
    close $w$ by adding shortest paths from $v_0$ to $s$ and from $t$ to $v_0$ in~$G$, where $s, t$ are the start and endpoints of $w$, respectively\nllabel{alg:carp:close}\; } \Return{$(\W, s)$}\;
  \caption{Algorithm for the proof of \cref{lem:carp}.}
  \label[algorithm]{alg:carp}
\end{algorithm*}

\paragraph{The algorithm.}
\cref{alg:carp} constructs an \MWCARP{} solution from an approximate
\MWRPP{} solution~$T$ containing all demand arcs and the depot~$v_0$.
In order to ensure that $T$~contains~$v_0$, \cref{alg:carp} assumes that
the input graph has a demand loop~$(v_0,v_0)$: If this loop is not
present, we can add it with zero cost.  Note that, while this does not
change the cost of an optimal solution, it might increase the number of
connected components in the subgraph induced by demand arcs by one.  To
compute an \MWCARP{} solution from~$T$, \cref{alg:carp} first computes a
\admsplit{}~$(\W,s)$ of~$T$. To each walk~$w_i\in\W$, it then adds a
shortest path from the end of~$w_i$ to the start of~$w_i$ via the
depot. It is not hard to check that \cref{alg:carp} indeed outputs a
feasible solution by using the properties of feasible splittings and the
fact that $T$~contains all demand arcs. 

\begin{remark}
\label{rem:split}
Instead of computing a feasible splitting of~\(T\) greedily,
\cref{alg:carp} could compute a splitting of~\(T\)
into pairwise non-overlapping segments
that provably minimizes the cost of the resulting MWCARP solution
\citep{Ulu85,BBLP06,Woe08,Jan93}.
Indeed, we will do so in our experiments in \cref{sec:experiments}.
For the analysis of the approximation factor, however,
the greedy splitting is sufficient and more handy,
since the analysis can exploit that two consecutive segments
of a feasible splitting serve more than \(Q\)~units of demand
(excluding, possibly, the last segment).
\end{remark}

\noindent
The remainder of this section is
devoted to the analysis of the solution cost,
thus proving the following proposition, which,
together with \cref{ourthm}\eqref{ourthm2},
yields \cref{ourthm}\eqref{ourthm3}.

\begin{proposition}\label[proposition]{lem:carp}
  Let $I = (G, v_0, c, d, \capact{})$ be an \MWCARP{} instance
and let $I'$~be the instance obtained from~$I$ by adding a zero-cost demand arc~$(v_0,v_0)$ if it is not present.

If \MWRPP{} is $\beta(\comps)$-approximable in $t(n)$~time, then \cref{alg:carp} applied to~$I'$ computes a $(8\beta(\comps + 1)+3)$\nobreakdash-\apxn{} for~$I$ in $t(C+1)+O(n^3)$~time.  Herein, $\comps$~is the number of connected components in~$G[R_d]$.
\end{proposition}

\noindent
The following lemma follows from the fact
that the concatenation of all vehicle tours
in any \MWCARP{} solution
yields an \MWRPP{} tour containing all demand arcs and the depot.

\begin{lemma}\label[lemma]{lem:shortsplitting}
  Let $I = (G, v_0, c, d, \capact{})$ be an \MWCARP{} instance with $(v_0,v_0)\in R_d$ and an optimal solution~$(\W^*,s^*)$.  The closed walk~$T$ and its \admsplit{}~$(\W,s)$ computed in \cref{alg:carp:basetour,alg:carp:split} of \cref{alg:carp} satisfy $\cost(\W) \leq \cost(T)\leq \beta(\comps)\cost(\W^*)$, where $\comps$~is the number of connected components in~$G[R_d]$.
\end{lemma}

\begin{proof}
  Consider an optimal solution~$(\W^*, s^*)$ to~$I$.  The closed walks
  in~$\W^*$ visit all arcs in~$R_d$.  Concatenating them to a closed
  walk~$T^*$ gives a feasible solution for the \MWRPP{}
  instance~$I'=(G,c,R_d)$ in \cref{alg:carp:drpinst} of \cref{alg:carp}.
  Moreover, $\cost(T^*)=\cost(\W^*)$.  Thus, we have
  $\cost(T) \leq \beta(\comps)\cost(T^*)$ in \cref{alg:carp:basetour}.
  Moreover, by \cref{def:admsplit}\eqref{nonoverlap}, one has
  $\cost(\W)\leq\cost(T)$.  This finally implies
  $\cost(\W)\leq\cost(T) \leq
  \beta(\comps)\cost(T^*)=\beta(\comps)\cost(\W^*)$ in
  \cref{alg:carp:split}.
\end{proof}

\noindent
For each~$w_i\in\W$, it remains to analyze the length of the shortest paths from~$v_0$ to~$w_i$ and from~$w_i$ to~$v_0$ added in \cref{alg:carp:close} of \cref{alg:carp}.  We bound their lengths in the lengths of an auxiliary walk~$\anf(w_i)$ from~$v_0$ to~$w_i$ and of an auxiliary walk~$\nde(w_i)$ from~$w_i$ to~$v_0$.  The auxiliary walks~$\anf(w_i)$ and~$\nde(w_i)$ consist of arcs of~$\W$, whose total cost is bounded by \cref{lem:shortsplitting}, and of arcs of an optimal solution~$(\W^*,s^*)$.  We show that, in total, the walks~$\anf(w_i)$ and~$\nde(w_i)$ for \emph{all}~$w_i\in\W$ use each subwalk of~$\W$ and~$\W^*$ at most a constant number of times. To this end, we group the walks in~$\W$ into consecutive pairs, for each of which we will be able to charge the cost of the auxiliary walks to a distinct vehicle tour of the optimal solution.
\begin{definition}[Consecutive pairing]
  For a feasible splitting~$(\W,s)$ with~$\W=(w_1,\dots,w_\ell)$,
  we call
  \[\Wpair{} :=
    \{(w_{2i-1}, w_{2i}) \mid
    i \in \{1, \ldots, \lfloor {\ell}/{2} \rfloor\}\}\]
  a \emph{consecutive pairing}.
\end{definition}
\noindent We can now show, by applying Hall's theorem \citep{Hal35}, that each pair traverses an arc from a \emph{distinct} tour of an optimal solution.
\begin{lemma}\label[lemma]{lem:assignment}
  Let $I=(G,v_0,c,d,\capact{})$~be an \MWCARP{} instance with an optimal solution~$(\W^*,s^*)$ and let $\Wpair{}$ be a consecutive pairing of some feasible splitting~$(\W,s)$.  Then, there is an injective map $\phi \colon \Wpair{} \to \W^*,(w_i,w_{i+1})\mapsto w^*$ such that $(s(w_i) \cup s(w_{i+1})) \cap s^*(w^*)\ne\emptyset$.  
\end{lemma}

\begin{proof}
  Define an undirected bipartite graph~$B$ with the partite
  sets~$\Wpair{}$ and~$\W^*$. A pair~$(w_i, w_{i+1}) \in \Wpair{}$
  and a closed walk~$w^* \in \W^*$ are adjacent in~$B$ if
  $(s(w_i) \cup s(w_{i+1})) \cap s^*(w^*)\ne\emptyset$.  We prove that
  $B$~allows for a matching that matches each vertex of~$\Wpair{}$ to
  some vertex in~$\W^*$.
  To this end, by Hall's theorem \citep{Hal35},
  it suffices to prove that, for each subset~$S \subseteq \Wpair{}$,
  it holds that $|N_B(S)| \geq |S|$,
  where $N_B(S):=\bigcup_{v\in S}N_B(v)$
  and $N_B(v)$~is the set of neighbors of a vertex~$v$ in~$B$.
  Observe that,
  by \cref{def:admsplit}\eqref{as:full} of feasible splittings,
  for each pair~$(w_i, w_{i+1}) \in \Wpair{}$, we have
  $d(s(w_i) \cup s(w_{i+1})) \geq \capact{}$.
  Since the pairs serve pairwise disjoint sets of demand arcs
  by \cref{def:admsplit}\eqref{as:partition},
  the pairs in~$S$ serve a total demand of at
  least~$\capact{}\cdot |S|$ in the closed walks~$N_B(S)\subseteq\W^*$.
  Since each closed walk in~$N_B(S)$ serves demand at most~$\capact{}$,
  the set~$N_B(S)$ is at least as large as~$S$, as required.
\end{proof}

\noindent
In the following, we fix
an arbitrary arc in~$(s(w_i) \cup s(w_{i+1})) \cap s^*(w^*)$
for each pair~\((w_i,w_{i+1})\in\Wpair{}\)
and call it the \emph{\reprarc{} arc} of~$(w_i,w_{i+1})$.
Informally, the auxiliary walks~$A(w_i)$, $Z(w_i)$ mentioned before are constructed as follows for each walk~$w_i$. To get from the endpoint of~$w_i$
to~$v_0$, walk along the closed walk~$T$ until traversing the first
\reprarc{} arc~$a$, then from the head of~$a$ to~$v_0$ 
follow the tour of~\(W^*\) containing~$a$.
To get
from~$v_0$ to~$w_i$, take the symmetric approach: walk backwards on~$T$
from the start point of~$w_i$ until traversing a \reprarc{} arc~\(a\)
and then follow the tour of~$\W^*$ containing~$a$.
The formal definition of the auxiliary walks~$\anf(w)$ and~$\nde(w)$
is given below and illustrated in \cref{fig:auxpaths}. %

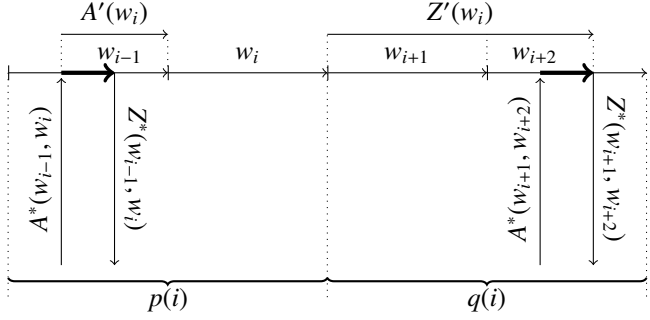
\begin{figure}
  \centering
  \begin{tikzpicture}[x=0.7cm,y=0.85cm]
    \draw[{|->}] (0,0) -- (3,0) node[above,pos=0.7] {$w_{i-1}$};
    \draw[{|->}] (3,0) -- (6,0) node[above,pos=0.5] {$w_{i}$};
    \draw[{|->}] (6,0) -- (9,0) node[above,pos=0.5] {$w_{i+1}$};
    \draw[{|->}] (9,0) -- (12,0) node[above,pos=0.3] {$w_{i+2}$};

    \draw[ultra thick, ->] (1,0)--(2,0);
    \draw[ultra thick, ->] (10,0)--(11,0);

    \draw [decorate,decoration={brace},thick] (6,-3.2)--(0,-3.2) node [midway,below] {$p(i)$};
    \draw [decorate,decoration={brace},thick] (12,-3.2)--(6,-3.2) node [midway,below] {$q(i)$};

    \draw[->, shorten >= 2pt] (1,-3) -- (1,0) node [midway,above,rotate=90] {$\anf^*(w_{i-1},w_{i})$};

    \draw[->, shorten >= 2pt] (10,-3) -- (10,0) node [midway,above,rotate=90] {$\anf^*(w_{i+1},w_{i+2})$};

    \draw[->] (2,0) -- (2,-3) node [midway,above,rotate=-90] {$\nde^*(w_{i-1},w_{i})$};

    \draw[->] (11,0) -- (11,-3) node [midway,above,rotate=-90] {$\nde^*(w_{i+1},w_{i+2})$};

    \draw[->] (1,0.6) -- (3,0.6) node [midway,above] {$\anf'(w_{i})$};

    \draw[->] (6,0.6) -- (11,0.6) node [midway,above] {$\nde'(w_i)$};

    \draw[dotted] (0,0)--(0,-3.5);
    \draw[dotted] (1,0)--(1,0.6);
    \draw[dotted] (1,0.6)--(1,0.6);
    \draw[dotted] (3,0)--(3,0.6);
    \draw[dotted] (6,-3.5)--(6,0.6);
    \draw[dotted] (9,0)--(9,0.6);
    \draw[dotted] (11,0)--(11,0.6);
    \draw[dotted] (12,0)--(12,-3.5);
  \end{tikzpicture}
  \caption{Illustration of \cref{def:auxpaths}.  Dotted lines are
    ancillary lines.  Thin arrows are walks. The braces along the
    bottom show a consecutive pairing of
    walks~$w_{i-1},\dots,w_{i+2}$.  Bold
    arcs are \reprarc{} arcs. Here, $p(i)$
    is exactly the pair that contains~$w_i$ and~$q(i)$ is the next
    pair.}
  \label{fig:auxpaths}
\end{figure}

\begin{definition}[Auxiliary walks]\label[definition]{def:auxpaths}
  Let $I=(G,v_0,c,d,\capact{})$~be an \MWCARP{} instance, $(\W^*,s^*)$ be an optimal solution, and $\Wpair{}$ be a consecutive pairing of some feasible splitting~$(\W,s)$ of a closed walk~$T$ containing all arcs~$R_d$ and~$v_0$, where $\W{} = (w_1, \ldots, w_\ell)$.

   Let $\phi \colon \Wpair{} \to \W^*$ be an injective map as in \cref{lem:assignment} and, for each pair~$(w_i, w_{i+1}) \in \Wpair$, let
  \begin{itemize}[$\anf^*(w_i,w_{i+1})$]
  \item[$\anf^*(w_i,w_{i+1})$] be a subwalk of~$\phi(w_i,w_{i+1})$ from~$v_0$ to the tail of the \reprarc{}  arc %
    of $(w_i,w_{i+1})$,
  \item[$\nde^*(w_i,w_{i+1})$] be a subwalk of~$\phi(w_i,w_{i+1})$ from the head of the \reprarc{} arc %
    of $(w_i,w_{i+1})$ to~$v_0$.
  \end{itemize}

\noindent  For each walk~$w_i\in\W$ with $i\geq 3$ (that is, $w_i$ is not in the first pair of~$\Wpair{}$), let
  \begin{itemize}[$\anf'(w_i)$]
  \item[$\jib{}$] be the index of the pair whose \reprarc{} arc is traversed first when walking~$T$ backwards starting from the starting point of~$w_i$,
  \item[$\anf'(w_i)$] be the subwalk of~$T$ starting at the end point of $\anf^*(w_{2\jib{}-1}, w_{2\jib{}})$ %
and ending at the start point of~$w_i$, and
  \item[$\anf(w_i)$] be the walk from~$v_0$ to the start point of~$w_i$ following first $\anf^*(w_{2\jib{}-1}, w_{2\jib{}})$ and then~$\anf'(w_i)$.
  \end{itemize}
\noindent  For each walk~$w_i\in\W$ with $i \leq \ell-3$ (that is, $w_i$ is not in the last pair of~$\Wpair{}$, where $w_\ell$ might not be in any pair if~$\ell$ is odd), let
  \begin{itemize}[$\anf'(w_i)$]
  \item[$\jia{}$] be the index of the pair whose \reprarc{} arc is traversed first when following~$T$ starting from the end point of~$w_i$,
  \item[$\nde'(w_i)$] be the subwalk of~$T$ starting at the end point of~$w_i$ and ending at the start point of~$\nde^*(w_{2\jia{}-1}, w_{2\jia{}})$, and
  \item[$\nde(w_i)$] be the walk from the end point of~$w_i$ to~$v_0$ following first~$\nde'(w_i)$ and then $\nde^*(w_{2\jia{}-1}, w_{2\jia{}})$.
  \end{itemize}
\end{definition}

\noindent We are now ready to prove \cref{lem:carp}, which also concludes our proof of \cref{ourthm}.

\begin{figure*}
  \centering
  \begin{tikzpicture}[x=0.75cm,y=0.85cm]
    \draw[{|->}] (-6,0) -- (-3,0) node[above,pos=0.6] {$w_{i-2}$};
    \draw[{|->}] (-3,0) -- (0,0) node[above,pos=0.5] {$w_{i-1}$};
    \draw[{|->}] (0,0) -- (3,0) node[above,pos=0.5] {$w_i$};
    \draw[{|->}] (3,0) -- (6,0) node[above,pos=0.5] {$w_{i+1}$};
    \draw[{|->}] (6,0) -- (9,0) node[above,pos=0.5] {$w_{i+2}$};
    \draw[{|->}] (9,0) -- (12,0) node[above,pos=0.4] {$w_{i+3}$};

    \draw[ultra thick, ->] (1,0)--(2,0);
    \draw[ultra thick, ->] (-5,0)--(-4,0);
    \draw[ultra thick, ->] (10,0)--(11,0);

    \draw [decorate,decoration={brace},thick] (0,-3.2)--(-6,-3.2) node [midway,below] {$p(i)$};
    \draw [decorate,decoration={brace},thick] (6,-3.2)--(0,-3.2);
    \draw [decorate,decoration={brace},thick] (12,-3.2)--(6,-3.2) node [midway,below] {$q(i)$};

    \draw[->, shorten >= 2pt] (-5,-3) -- (-5,0) node [midway,above,rotate=90] {$\anf^*(w_{i-2},w_{i-1})$};

    \draw[->, shorten >= 2pt] (1,-3) -- (1,0) node [midway,above,rotate=90] {$\anf^*(w_{i},w_{i+1})$};

    \draw[->, shorten >= 2pt] (10,-3) -- (10,0) node [midway,above,rotate=90] {$\anf^*(w_{i+2},w_{i+3})$};

    \draw[->] (-4,0) -- (-4,-3) node [midway,above,rotate=-90] {$\nde^*(w_{i-2},w_{i-1})$};

    \draw[->] (2,0) -- (2,-3) node [midway,above,rotate=-90] {$\nde^*(w_{i},w_{i+1})$};

    \draw[->] (11,0) -- (11,-3) node [midway,above,rotate=-90] {$\nde^*(w_{i+2},w_{i+3})$};

    \draw[->] (-5,0.6) -- (-0.05,0.6) node [pos=0.6,above] {$\anf'(w_i)$};
    \draw[->] (1,1.7) -- (3,1.7) node [midway,above] {$\anf'(w_{i+1})$};
    \draw[->] (1,2.4) -- (6,2.4) node [midway,above] {$\anf'(w_{i+2})$};
    \draw[->] (1,1.5) -- (9,1.5) node [midway,above] {$\anf'(w_{i+3})$};

    \draw[->] (-3,1.3) -- (2,1.3)  node [midway,above] {$\nde'(w_{i-2})$};

    \draw[->] (0,0.6) -- (2,0.6)  node [midway,above] {$\nde'(w_{i-1})$};

    \draw[->] (3,0.6) -- (11,0.6) node [midway,above] {$\nde'(w_i)$};

    \draw[dotted] (0,0.6)--(0,-3.5);
    \draw[dotted] (-5,0.6)--(-5,0);
    \draw[dotted] (-3,1.2)--(-3,0);
    \draw[dotted] (1,0)--(1,0.8);
    \draw[dotted] (1,1.1)--(1,2.4);
    \draw[dotted] (2,0)--(2,1.3);
    \draw[dotted] (3,0)--(3,1.6);
    \draw[dotted] (6,-3.5)--(6,2.4);
    \draw[dotted] (9,0)--(9,1.5);
    \draw[dotted] (11,0)--(11,0.6);
    \draw[dotted] (-6,0)--(-6,-3.5);
    \draw[dotted] (12,0)--(12,-3.5);
  \end{tikzpicture}
  \caption{Illustration of the situation in which a maximum number of five different walks in~$\W$ traverse the same \reprarc{} arc (the bold arc of~$w_i$) in their respective auxiliary walks.  }
  \label{fig:auxpaths2}
\end{figure*}
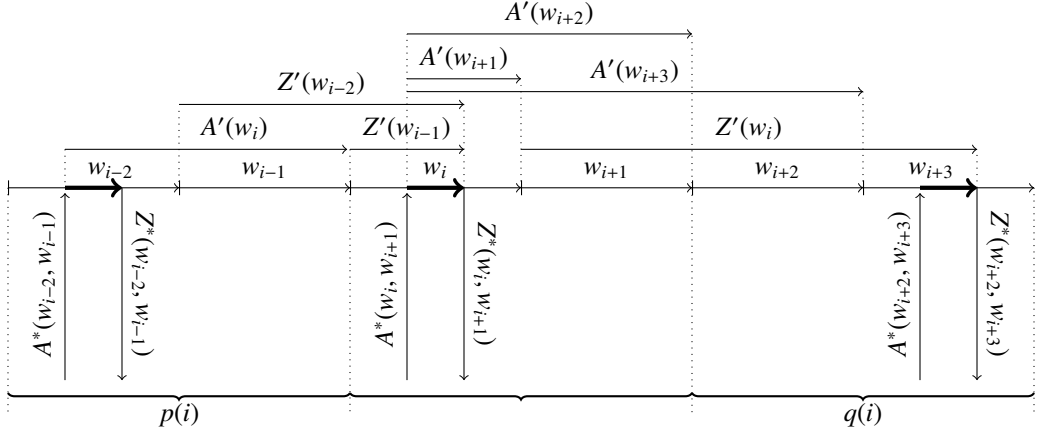

\begin{proof}[Proof of \cref{lem:carp}]
  Let $I = (G, v_0, c, d, \capact{})$ be an \MWRPP{} instance and~$(\W^*,s^*)$ be an optimal solution.  If there is no demand arc~$(v_0,v_0)$ in~$I$, then we add it with zero cost in order to make \cref{alg:carp} applicable.  This clearly does not change the cost of an optimal solution but may increase the number of connected components of~$G[R_d]$ to~$\comps+1$.

  In \cref{alg:carp:basetour,alg:carp:split}, \cref{alg:carp}  computes a tour~$T$ and its \admsplit{}~$(\W,s)$,
  which works in $t(C+1)+O(n^3)$~time by \cref{ourthm}\eqref{ourthm2}.
  Denote $\W=(w_1,\dots,w_\ell)$. The solution returned by \cref{alg:carp} consists, for each~$1\leq i\leq\ell$, of a tour starting in~$v_0$, following a shortest path to the starting point of~$w_i$, then~$w_i$, and a shortest path back to~$v_0$.

  For $i\geq 3$, the shortest path from~$v_0$ to the starting point of~$w_i$ has length at most~$\cost(\anf(w_i))$. For $i\leq \ell-3$, the shortest path from the end point of~$w_i$ to~$v_0$ has length at most~$\cost(\nde(w_i))$. This amounts to $\sum_{i=3}^{\ell}\cost(\anf(w_i))+\sum_{i=1}^{\ell-3}\cost(\nde(w_i))$. To bound the costs of the shortest paths attached to~$w_i$ for $i\in\{1,2,\ell-2,\ell-1,\ell\}$, observe the following. For each~$i\in\{1,2\}$, the shortest paths from~$v_0$ to the start point of~$w_i$ and from the end point of~$w_{\ell-i}$ to~$v_0$ together have length at most~$\cost(T)$. The shortest path from the end point of~$w_{\ell}$ to~$v_0$ has length at most~$\cost(T)-\cost(\W{})$.
  Thus, the solution returned by \cref{alg:carp} has cost at most%
  \begin{align*}\allowdisplaybreaks
    &\sum_{i=1}^\ell \cost(w_i)+\sum_{i=3}^{\ell}\cost(\anf(w_i))+\sum_{i=1}^{\ell-3}\cost(\nde(w_i))+3\cost(T)-\cost(\W{})\\
={}&\sum_{i=3}^{\ell}\cost(\anf(w_i))+\sum_{i=1}^{\ell-3}\cost(\nde(w_i))+3\cost(T)\\={}&3\cost(T)+{}\\
    &+\sum_{i=3}^{\ell}\cost(\anf^*(w_{2\jib{}-1}, w_{2\jib{}}))+\sum_{i=1}^{\ell-3}\cost(\nde^*(w_{2\jia{}-1}, w_{2\jia{}}))+{}\tag{S1}\label{s2}\\
    &+\sum_{i=3}^{\ell}\cost(\anf'(w_i))+\sum_{i=1}^{\ell-3}\cost(\nde'(w_i)).\tag{S2}\label{s1}
  \end{align*}
 Observe that, for a fixed~$i$, one has $p(i)=p(j)$ only for~$j\leq i+2$ and $q(i)=q(j)$ only for~$j\geq i-2$.
Moreover, by \cref{lem:assignment} and \cref{def:auxpaths},
if $p(i)\neq p(j)$, then $\anf^*(w_{2p(i)-1}, w_{2p(i)})$ and $\anf^*(w_{2p(j)-1}, w_{2p(j)})$ are subwalks of distinct walks of~$\W^*$. Similarly, $\nde^*(w_{2q(i)-1}, w_{2q(i)})$ and $\nde^*(w_{2q(j)-1}, w_{2q(j)})$ are subwalks of distinct walks of~$\W^*$ if~$q(i)\neq q(j)$.  Hence, sum~\eqref{s2} counts each arc of~$\W^*$ at most three times and is therefore bounded from above by~$3\cost(\W^*)$.

Now, for a walk~$w_i$, let $\mathcal{A}_i$~be the set of walks~$w_j$
such that any arc $a$ of~$w_i$~is contained in~$\anf'(w_j)$ and let
$\mathcal{Z}_i$~be the set of walks such that any arc~$a$ of~$w_i$~is
contained in~$\nde'(w_j)$. Observe that~$\anf'(w_j)$ and~$\nde'(w_j)$
cannot completely contain two walks of the same pair of the consecutive
pairing~$\Wpair{}$ of~$\W$ since, by \cref{lem:assignment}, each pair
has a \reprarc{} arc and~$\anf'(w_j)$ and~$\nde'(w_j)$ both stop after
traversing a \reprarc{} arc. Hence, the walks
in~$\mathcal{A}_i\cup \mathcal{Z}_i$ can be from at most three pairs
of~$\Wpair{}$: the pair containing~$w_i$ and the two neighboring
pairs. Finally, observe that~$w_i$ itself is not contained
in~$\mathcal{A}_i\cup \mathcal{Z}_i$. Thus,
$\mathcal{A}_i\cup \mathcal{Z}_i$ contains at most five walks
(\cref{fig:auxpaths2} shows a worst-case example).  Therefore, sum~\eqref{s1}
counts every arc of~$\W{}$ at most five times and is bounded
from above by~$5\cost(\W{})$.

Thus, \cref{alg:carp} returns a solution of cost $3\cost(T)+5\cost(\W{})+3\cost(\W^*)$ which, by \cref{lem:shortsplitting}, is at most
$8\cost(T)+3\cost(\W^*)\leq 8\beta(\comps+1)\cost(\W^*)+3\cost(\W^*)\leq (8\beta(\comps+1)+3)\cost(\W^*)$.
\end{proof}

\section{Experiments}
\label{sec:experiments}

Our approximation algorithm for MWCARP is one of many
``route first, cluster second''-approaches,
which was first applied to CARP by \citet{Ulu85}
and led to constant-factor approximations for 
the undirected CARP \citep{Woe08,Jan93}.
Notably, 
\citet{BBLP06} implemented \citeauthor{Ulu85}'s heuristic \citep{Ulu85}
for the mixed CARP
by computing the base tour using path scanning heuristics.
Our experimental evaluation will show that
\citeauthor{Ulu85}'s heuristic can be substantially improved
by computing the base tour
using our \cref{ourthm}\eqref{ourthm2}.

For the evaluation,
we use the \texttt{mval} and \texttt{lpr} benchmark sets
of \citet{BBLP06}
for the mixed (but non-windy) CARP
and the \texttt{egl-large} benchmark set
of \citet{BE08} for the (undirected) CARP.
We chose these benchmark sets because relatively good lower bounds
to compare with are known \citep{GMP10,BI15}.
Moreover,
the \texttt{egl-large} set is of particular interest
since it contains large instances derived from real road networks
and the \texttt{mval} and \texttt{lpr} sets
are of particular interest
since \citet{BBLP06} used them to evaluate
their variant of \citeauthor{Ulu85}'s heuristic \citep{Ulu85},
which is very similar to our algorithm.

In the following, \cref{sec:impl} describes some
heuristic enhancements of our algorithm,
\cref{sec:expres} interprets our experimental results,
and \cref{sec:ob} describes an approach to transform
instances of existing benchmark sets
into instances whose positive-demand arcs induce a moderate number
of connected components.

\subsection{Implementation details}
\label{sec:impl}

Since our main goal is evaluating the solution quality
rather than the running time of our algorithm,
we sacrificed speed for simplicity
and implemented it in Python.\footnote{Source code
available at \url{http://gitlab.com/rvb/mwcarp-approx}}
Thus, the running time of our implementation is not competitive
to the implementations by \citet{BBLP06} and \citet{BE08}.%
\footnote{We do not provide running time measurements
since we processed many instances in parallel,
which does not yield reliable measurements.}
However, it is clear that a careful implementation
of our algorithm in C++
will yield competitive running times:
The most expensive steps of our algorithm
are the Floyd-Warshall all-pair shortest path algorithm \citep{Flo62},
which is also used by \citet{BBLP06} and \citet{BE08},
and the computation of an uncapacitated minimum-cost flow,
algorithms for which are contained
in highly optimized C++ libraries
like LEMON.\footnote{\url{http://lemon.cs.elte.hu/}}

In the following, we describe heuristic improvements
over the algorithms presented in \cref{sec:rpp,sec:carp},
which were described there so as
to conveniently prove upper bounds
rather than focusing on good solutions.

\subsubsection{Joining connected components}
\label{sec:connectbf}
We observed that,
in all but one instance of the \texttt{egl-large}, \texttt{lpr},
and \texttt{mval} benchmark sets,
the set of positive-demand arcs induce only one connected component.
Therefore, connecting them is usually not necessary and
the call to \cref{alg:eulerian-rp} in \cref{alg:drpp}
can be skipped completely.
If not, then, contrary to the description of \cref{alg:eulerian-rp},
we do \emph{not} arbitrarily select one vertex
from each connected component and join them using an approximate
\TATSP{} tour as in \cref{alg:eulerian-rp} or using an optimal
\TATSP{} tour as for \cref{cor:fpt-apx}.

Instead, using brute force,
we try all possibilities of choosing one vertex
from each connected component and connecting them using a cycle
and choose the cheapest variant.
If the positive-demand arcs induce \(C\)~connected components,
then this takes \(O(n^C\cdot C!+n^3)\)~time in an \(n\)-vertex graph.
That is, for \(C\leq 3\), implementing \cref{alg:eulerian-rp}
in this way does not increase its asymptotic time complexity.

\subsubsection{Choosing service direction}
\label{sec:serdir}
The instances in the \texttt{egl-large}, \texttt{lpr}, and \texttt{mval}
benchmark sets are not windy.
Thus, as pointed out in \cref{rem:orientation},
when computing the MWRPP base tour,
we are free to choose whether to replace
a required undirected edge~\(\{u,v\}\)
by a required arc~\((u,v)\) or
a required arc~\((v,u)\)
(and adding the opposite non-required arc)
without increasing the approximation factor
in \cref{ourthm}\eqref{ourthm2}.

We thus implemented several heuristics for choosing what we call
the \emph{service direction} of the undirected edge~\(\{u,v\}\).
Some of these heuristics choose the service direction independently
for each undirected edge, similarly to \citet{CMR00},
others choose it for whole undirected paths and cycles,
similarly to \citet{MA05}.

We now describe these heuristics in detail.  To this end, 
let \(G\)~denote our input graph and
\(R\)~be the set of required arcs.

\begin{itemize}[EO(R)]
\item[EO(R)] assigns one of the two possible service directions
  to each undirected edge uniformly at random.
\item[EO(P)] replaces each undirected edge~\(\{u,v\}\in R\)
by an arc~\((u,v)\in R\) if \(\balance_{G[R]}(v)<\balance_{G[R]}(u)\),
by an arc~\((v,u)\in R\) if \(\balance_{G[R]}(v)>\balance_{G[R]}(u)\),
and chooses a random service direction otherwise.
\item[EO(S)] randomly chooses one endpoint~\(v\)
  of each undirected edge~\(\{u,v\}\in R\) and replaces
  it by an arc~\((u,v)\in R\) if \(\balance_{G[R]}(v)<0\)
  and by~\((v,u)\in R\) otherwise.
\end{itemize}
Herein, ``EO'' is for ``edge orientation''.
The ``R'' in parentheses is for ``random'',
the ``P'' for ``pair''
(since it levels the balances of pairs of vertices),
and the ``S'' is for ``single''
(since it minimizes~\(|\balance(v)|\) of a single
random endpoint~\(v\) of the edge).

In addition, we experiment with three heuristics
that do not orient independent edges but long undirected paths.
Herein, the aim is that a vehicle will be able to
serve all arcs resulting from such a path
in one run.

First, the heuristics repeatedly search for undirected cycles in~\(G[R]\)
and replace them by directed cycles in~\(R\).
When no undirected cycle is left, then
the undirected edges of~\(G[R]\)~form a forest.
The heuristics then repeatedly search for a
longest undirected path in~\(G[R]\)
and choose its service direction as follows.

\begin{itemize}[PO(R)]
\item[PO(R)] assigns the service direction randomly.
\item[PO(P)] assigns the service direction by leveling the balance
of the endpoints of the path, analogously to EO(P).
\item[PO(S)] assigns the service direction so as to minimize
\(|\balance(v)|\) for a random endpoint~\(v\) of the path,
analogously to EO(S).
\end{itemize}
\noindent
Generally, we observed that these heuristics
first find three or four long paths
with lengths from~5 up to~15. 
Then, the length of the found paths quickly decreases:
In most instances,
at least half of all found paths have length one,
at least 3/4 of all found paths have length at most two.

We now present experimental results for each of these six heuristics.

\subsubsection{Tour splitting} 
\label{sec:optsplit}
As pointed out in \cref{rem:split},
the MWRPP base tour initially computed in \cref{alg:carp}
can be split into pairwise non-overlapping subsequences
so as to minimize the total cost of the resulting vehicle tours.
To this end, we apply an approach of
\citet{Bea83} and \citet{Ulu85},
which by now can be considered folklore \citep{BBLP06,Woe08,Jan93} and works as follows.

Denote the positive-demand arcs on the MWRPP base tour
as a sequence~\(a_1,\dots,a_\ell\).
To compute the optimal splitting,
we create an auxiliary graph with the vertices~\(1,\dots,\ell+1\).
Between each pair~\((i,j)\) of vertices,
there is an edge whose weight is the cost for serving
all arcs~\(a_i,a_{i+1},\dots,a_{j-1}\) in this order
using one vehicle.
That is, its cost is~\(\infty\) if the demands of the arcs
in this segment exceed the vehicle capacity~\(Q\) and
otherwise it is the cost
for going from the depot~\(v_0\) to the tail of~\(a_i\),
serving arcs~\(a_i\) to~\(a_{j-1}\),
and returning from the head of~\(a_{j-1}\) to the depot~\(v\).
Then, a shortest path from vertex~\(1\) to~\(\ell+1\)
in this auxiliary graph
gives an optimal splitting of the MWRPP base tour
into mutually non-overlapping subsequences.

Additionally, we implemented a trick of \citet{BBLP06}
that takes into account that a vehicle may serve
a segment~\(a_i,\dots,a_k,a_{k+1}\dots,a_{j-1}\)
by going to the tail of~\(a_{k+1}\),
serving arcs~\(a_{k+1}\) to~\(a_{j-1}\),
going from the head of~\(a_{j-1}\) to the tail of~\(a_i\),
serving arcs~\(a_i\) to~\(a_k\),
and finally returning from the head of~\(a_k\) to the depot~\(v_0\).
Our implementation tries all such~\(k\)
and assigns the cheapest resulting cost
to the edge between the pair~\((i,j)\) of vertices
in the auxiliary graph.

Of course one could compute the optimal order
for serving the arcs of a segment~\(a_i,\dots,a_{j-1}\)
from the depot~\(v_0\),
but this would again be the NP-hard DRPP.

\subsection{Experimental results}
\label{sec:expres}

Our experimental results
for the \texttt{lpr}, \texttt{mval}, and \texttt{egl-large} instances
are presented in \cref{tab:lpr,tab:egl}.
We grouped the results for the \texttt{lpr} and \texttt{mval} instances
into one table and subsection
since our conclusions about them are very similar.
We explain and interpret the tables in the following.

\begin{table*}[p]
  \caption{Known results \citep{BBLP06,GMP10} and our results
    for the \texttt{lpr} and \texttt{mval} instances.
  See \cref{sec:lpr-mval} for a description of the table.  The best polynomial-time computed upper bound
is written in boldface, the second best is underlined,
names of instances solved optimally by our algorithms
are also written in boldface.}
  \label{tab:lpr}
  \centering
  \begin{tabular}{l|rr|rrr|rrrrrr}
\toprule
& \multicolumn{5}{c|}{Known results} & \multicolumn{6}{c}{Our results}   \\
Instance           & LB       & UB       & PSRC     & IM                 & IURL             & EO(R)             & EO(P)             & EO(S)             & PO(R)             & PO(P)             & PO(S)             \\
\midrule
\bestsol{lpr-a-01} & 13\,484  & 13\,484  & 13\,600  & 13\,597            & \secsol{13\,537} & \bestsol{13\,484} & \bestsol{13\,484} & \bestsol{13\,484} & \bestsol{13\,484} & \bestsol{13\,484} & \bestsol{13\,484} \\
lpr-a-02           & 28\,052  & 28\,052  & 29\,094  & 28\,377            & 28\,586          & \bestsol{28\,225} & 28\,381           & 28\,356           & \secsol{28\,239}  & 28\,381           & 28\,356           \\
lpr-a-03           & 76\,115  & 76\,155  & 79\,083  & 77\,331            & 78\,151          & 77\,019           & \bestsol{76\,783} & 76\,964           & 76\,951           & \bestsol{76\,783} & \secsol{76\,820}  \\
lpr-a-04           & 126\,946 & 127\,352 & 133\,055 & \bestsol{128\,566} & 131\,884         & 130\,470          & \secsol{130\,137} & 130\,255          & 130\,198          & 130\,171          & 130\,186          \\
lpr-a-05           & 202\,736 & 205\,499 & 215\,153 & \bestsol{207\,597} & 212\,167         & 210\,328          & \secsol{209\,980} & 210\,265          & 210\,235          & 210\,139          & 210\,344          \\
\bestsol{lpr-b-01} & 14\,835  & 14\,835  & 15\,047  & 14\,918            & \secsol{14\,868} & 14\,869           & 14\,869           & \bestsol{14\,835} & \bestsol{14\,835} & \bestsol{14\,835} & \bestsol{14\,835} \\
lpr-b-02           & 28\,654  & 28\,654  & 29\,522  & 29\,285            & 28\,947          & 28\,749           & \bestsol{28\,689} & \bestsol{28\,689} & 28\,757           & 28\,790           & \secsol{28\,727}  \\
lpr-b-03           & 77\,859  & 77\,878  & 80\,017  & 80\,591            & 79\,910          & \bestsol{78\,428} & 78\,745           & 78\,853           & \secsol{78\,645}  & 78\,810           & 78\,743           \\
lpr-b-04           & 126\,932 & 127\,454 & 133\,954 & \bestsol{129\,449} & 132\,241         & \secsol{130\,024} & \secsol{130\,024} & \secsol{130\,024} & 130\,076          & \secsol{130\,024} & \secsol{130\,024} \\
lpr-b-05           & 209\,791 & 211\,771 & 223\,473 & \bestsol{215\,883} & 219\,702         & 217\,024          & 216\,769          & \secsol{216\,459} & 217\,079          & 216\,639          & 216\,659          \\
lpr-c-01           & 18\,639  & 18\,639  & 18\,897  & 18\,744            & \secsol{18\,706} & 18\,943           & \bestsol{18\,695} & 18\,732           & 18\,708           & 18\,752           & 18\,752           \\
lpr-c-02           & 36\,339  & 36\,339  & 36\,929  & \bestsol{36\,485}  & 36\,763          & 37\,177           & \secsol{36\,649}  & 36\,856           & 36\,723           & 36\,711           & 36\,662           \\
lpr-c-03           & 111\,117 & 111\,632 & 115\,763 & \bestsol{112\,462} & 114\,539         & 115\,399          & 114\,438          & 114\,888          & 114\,336          & 114\,335          & \secsol{114\,290} \\
lpr-c-04           & 168\,441 & 169\,254 & 174\,416 & \bestsol{171\,823} & 173\,161         & 174\,088          & \secsol{172\,089} & 172\,902          & 172\,637          & 172\,172          & 172\,365          \\
lpr-c-05           & 257\,890 & 259\,937 & 268\,368 & \bestsol{262\,089} & 266\,058         & 266\,637          & \secsol{263\,989} & 264\,947          & 264\,911          & 264\,263          & 264\,665          \\
\midrule
Instance           & LB       & UB       & PSRC     & IM                 & IURL             & EO(R)             & EO(P)             & EO(S)             & PO(R)             & PO(P)             & PO(S)             \\
\midrule
\bestsol{mval1A} & 230 & 230 & 243  & 243           & \secsol{231}  & 245   & \bestsol{230} & 238           & 234           & 239           & 234           \\
mval1B           & 261 & 261 & 314  & \bestsol{276} & 292           & 298   & \secsol{285}  & \secsol{285}  & 307           & 307           & 307           \\
mval1C           & 309 & 315 & 427  & \bestsol{352} & \secsol{357}  & 367   & 362           & 362           & 367           & 372           & 370           \\
\bestsol{mval2A} & 324 & 324 & 409  & 360           & 374           & 397   & \secsol{353}  & \bestsol{324} & 369           & 369           & 368           \\
mval2B           & 395 & 395 & 471  & \bestsol{407} & 434           & 431   & \secsol{424}  & \secsol{424}  & \secsol{424}  & \secsol{424}  & \secsol{424}  \\
mval2C           & 521 & 526 & 644  & \bestsol{560} & 601           & 621   & 622           & 592           & \secsol{600}  & 624           & 594           \\
mval3A           & 115 & 115 & 133  & \bestsol{119} & 128           & 131   & 129           & 125           & 122           & \secsol{121}  & \secsol{121}  \\
mval3B           & 142 & 142 & 162  & 163           & 150           & 151   & \secsol{148}  & 151           & 149           & \bestsol{147} & \bestsol{147} \\
mval3C           & 166 & 166 & 191  & \bestsol{174} & 192           & 194   & 190           & \secsol{189}  & 194           & 200           & 200           \\
mval4A           & 580 & 580 & 699  & 653           & 684           & 648   & \bestsol{622} & \secsol{645}  & 651           & 647           & 647           \\
mval4B           & 650 & 650 & 775  & 693           & 737           & 709   & 687           & 709           & 690           & \bestsol{674} & \secsol{682}  \\
mval4C           & 630 & 630 & 828  & \bestsol{702} & 740           & 750   & 721           & 736           & \secsol{714}  & 722           & 722           \\
mval4D           & 746 & 770 & 1015 & \bestsol{810} & 905           & 875   & 871           & \secsol{852}  & 872           & 879           & 870           \\
mval5A           & 597 & 597 & 733  & 686           & 683           & 672   & \secsol{619}  & 652           & \bestsol{614} & 649           & 644           \\
mval5B           & 613 & 613 & 718  & 677           & 677           & 687   & 662           & 685           & \bestsol{653} & \bestsol{653} & \secsol{654}  \\
mval5C           & 697 & 697 & 809  & \bestsol{743} & 811           & 788   & \secsol{773}  & 778           & 783           & 804           & 783           \\
mval5D           & 719 & 739 & 883  & \bestsol{821} & 855           & 859   & 840           & 854           & 845           & 840           & \secsol{836}  \\
mval6A           & 326 & 326 & 392  & 370           & 367           & 348   & \secsol{347}  & 348           & \bestsol{344} & 351           & 350           \\
mval6B           & 317 & 317 & 406  & 346           & 354           & 345   & \bestsol{331} & 354           & 351           & \secsol{343}  & 347           \\
mval6C           & 365 & 371 & 526  & \bestsol{402} & 444           & 455   & \secsol{435}  & \secsol{435}  & 461           & 454           & 454           \\
mval7A           & 364 & 364 & 439  & \bestsol{381} & 390           & 428   & \secsol{386}  & 411           & 404           & 398           & 398           \\
mval7B           & 412 & 412 & 507  & 470           & 491           & 474   & \bestsol{435} & 463           & 460           & 460           & \secsol{454}  \\
mval7C           & 424 & 426 & 578  & \bestsol{451} & 504           & 507   & 474           & 483           & 489           & \secsol{482}  & \secsol{482}  \\
mval8A           & 581 & 581 & 666  & 639           & 651           & 648   & \secsol{635}  & \secsol{635}  & 639           & \bestsol{627} & 641           \\
mval8B           & 531 & 531 & 619  & \bestsol{568} & 611           & 616   & \secsol{582}  & 592           & 596           & 598           & 600           \\
mval8C           & 617 & 638 & 842  & \bestsol{718} & 762           & 799   & 737           & \secsol{729}  & 776           & 764           & 779           \\
mval9A           & 458 & 458 & 529  & 500           & 514           & 503   & \bestsol{486} & 493           & 496           & \secsol{490}  & 498           \\
mval9B           & 453 & 453 & 552  & 534           & \bestsol{502} & 518   & 504           & \secsol{503}  & \secsol{503}  & 523           & 506           \\
mval9C           & 428 & 429 & 529  & 479           & 498           & 509   & \bestsol{468} & 488           & 485           & 479           & \secsol{474}  \\
mval9D           & 514 & 520 & 695  & \bestsol{575} & 622           & 627   & \secsol{603}  & 610           & 612           & 613           & 608           \\
mval10A          & 634 & 634 & 735  & 710           & 705           & 669   & 663           & 661           & 667           & \bestsol{658} & \secsol{659}  \\
mval10B          & 661 & 661 & 753  & 717           & 714           & 708   & \bestsol{687} & \secsol{693}  & 703           & 703           & 698           \\
mval10C          & 623 & 623 & 751  & \bestsol{680} & 714           & 709   & 689           & 697           & 698           & 695           & \secsol{687}  \\
mval10D          & 643 & 649 & 847  & \bestsol{706} & 760           & 778   & 739           & 763           & 775           & 743           & \secsol{722}  \\
\bottomrule
\end{tabular}
\end{table*}

\begin{table*}
  \centering
\caption{Known results \citep{BE08,BI15} and our results
    for the \texttt{egl-large} instances.
  See \cref{sec:egl} for a description of the table.  The best polynomial-time computed upper bound
is written in boldface.}
  \begin{tabular}{l|rrr|rrrrrr}
\toprule
& \multicolumn{3}{c|}{Known results} & \multicolumn{6}{c}{Our results}   \\
Instance & LB          & UB          & PS          & EO(R)       & EO(P)       & EO(S)       & PO(R)                 & PO(P)                 & PO(S)                             \\
\midrule
egl-g1-A & 976\,907    & 1\,049\,708 & 1\,318\,092 &  1\,258\,206&1\,181\,928&1\,209\,108&1\,153\,029&1\,158\,233&\bestsol{1\,141\,457}\\
egl-g1-B & 1\,093\,884 & 1\,140\,692 & 1\,483\,179 &  1\,367\,979&1\,306\,521&1\,328\,250&1\,293\,095&1\,308\,350&\bestsol{1\,297\,606}\\
egl-g1-C & 1\,212\,151 & 1\,282\,270 & 1\,584\,177 &  1\,523\,183&1\,456\,305&1\,463\,009&1\,432\,281&\bestsol{1\,424\,722}&1\,430\,841\\
egl-g1-D & 1\,341\,918 & 1\,420\,126 & 1\,744\,159 &  1\,684\,343&1\,609\,822&1\,609\,537&1\,586\,294&1\,601\,588&\bestsol{1\,580\,634}\\
egl-g1-E & 1\,482\,176 & 1\,583\,133 & 1\,841\,023 &  1\,829\,244&1\,769\,977&1\,780\,089&1\,716\,612&\bestsol{1\,748\,308}&1\,755\,700\\
egl-g2-A & 1\,069\,536 & 1\,129\,229 & 1\,416\,720 &  1\,372\,177&1\,276\,871&1\,304\,618&1\,263\,263&\bestsol{1\,249\,293}&1\,255\,120\\
egl-g2-B & 1\,185\,221 & 1\,255\,907 & 1\,559\,464 &  1\,517\,245&1\,410\,385&1\,449\,553&1\,398\,162&1\,405\,916&\bestsol{1\,404\,533}\\
egl-g2-C & 1\,311\,339 & 1\,418\,145 & 1\,704\,234 &  1\,661\,596&1\,594\,147&1\,597\,266&1\,538\,036&\bestsol{1\,532\,913}&1\,544\,214\\
egl-g2-D & 1\,446\,680 & 1\,516\,103 & 1\,918\,757 &  1\,812\,309&1\,728\,840&1\,741\,351&1\,695\,333&\bestsol{1\,694\,448}&1\,704\,080\\
egl-g2-E & 1\,581\,459 & 1\,701\,681 & 1\,998\,355 &  1\,962\,802&1\,883\,953&1\,908\,339&\bestsol{1\,851\,436}&1\,861\,134&1\,861\,469\\
\bottomrule
  \end{tabular}
  \label{tab:egl}
\end{table*}

\subsubsection{Results for the lpr and mval instances}
\label{sec:lpr-mval}
\cref{tab:lpr} presents known results and our results
for the \texttt{lpr} and \texttt{mval} instances.
Each column for our results was obtained by running
our algorithm with
the corresponding service direction heuristic described in \cref{sec:serdir}
on each instance 20~times and reporting the best result.
The number 20 has been chosen so that our results are comparable with
those of \citet{BBLP06}, who
used the same number of runs
for their path scanning heuristic (column PSRC)
and their ``route first, cluster second'' heuristic (column IURL),
which computes the base tour using a path scanning heuristic
and then splits it using all tricks described in \cref{sec:optsplit}.
Columns LB and UB report the best
lower and upper bounds computed by \citet{BBLP06} and \citet{GMP10}
(usually not using polynomial-time algorithms).
Finally, column IM shows the result that \citet{BBLP06} obtained
using an improved variant of the ``augment and merge'' heuristic
due to \citet{GW81}.

\cref{tab:lpr} shows
that our algorithm with the EO(S) service direction heuristic
solved three instances optimally,
which other polynomial-time heuristics did not.
The EO(P) heuristic solved one instance optimally,
which also other polynomial-time heuristics did not.
Moreover, whenever no variant of our algorithm finds the best result,
then some variant yields the second best.
It is outperformed only by IM in 26 out of 49 instances
and by IURL in only one instance.
Apparently, our algorithm outperforms PSRC and IURL.
Notably, IURL differs from our algorithm only
in computing the base tour heuristically
instead of using our \cref{ourthm}\eqref{ourthm2}.
Thus, ``route first, cluster second'' heuristics
seem to benefit from computing the base tour
using our MWRPP approximation algorithm.

Remarkably,
when our algorithm yields the best result
using one of
the service direction heuristics described in \cref{sec:serdir},
then usually other service direction heuristics also find
the best or at least the second best solution.
Thus, the choice of the service direction heuristic
does not play a strong role.  
Indeed, we also experimented with repeating our algorithm
20~times on each instance,
each time choosing the service direction heuristic randomly.
The results come close to choosing the best heuristic for each instance.

\subsubsection{Results for the egl-large instances}
\label{sec:egl}
\cref{tab:egl} reports known results and our results
for the \texttt{egl-large} benchmark set.
Again, each column for our results was obtained
by running our algorithm with the corresponding
service direction heuristic
described in \cref{sec:serdir} on each instance 20~times.
The column LB reports lower bounds by \citet{BI15},
the column UB shows the upper bound that \citet{BE08}
obtained using their tabu-search algorithm
(which generally does not run in polynomial time).
The column PS shows the cost of the initial solution that \citet{BE08}
computed for their tabu-search algorithm
using a path scanning heuristic.
\citet{BE08} implemented several polynomial-time heuristics
for computing these initial solution.
Among them, ``route first, cluster second'' approaches
and ``augment and merge'' heuristics.
In their work, path scanning yielded the best initial solutions.
In \cref{tab:egl}, we see that
our algorithm clearly outperforms it.
Moreover, we see that especially our PO service direction heuristics
are successful.
This is because the \texttt{egl-large} instances
are undirected and, thus,
contain many cycles consisting of undirected positive-demand
arcs that can be directed by our PO heuristics
without increasing the imbalance of vertices.

\begin{figure*}[tb]
  \begin{subfigure}[b]{0.4\textwidth}
    \includegraphics[width=\textwidth]{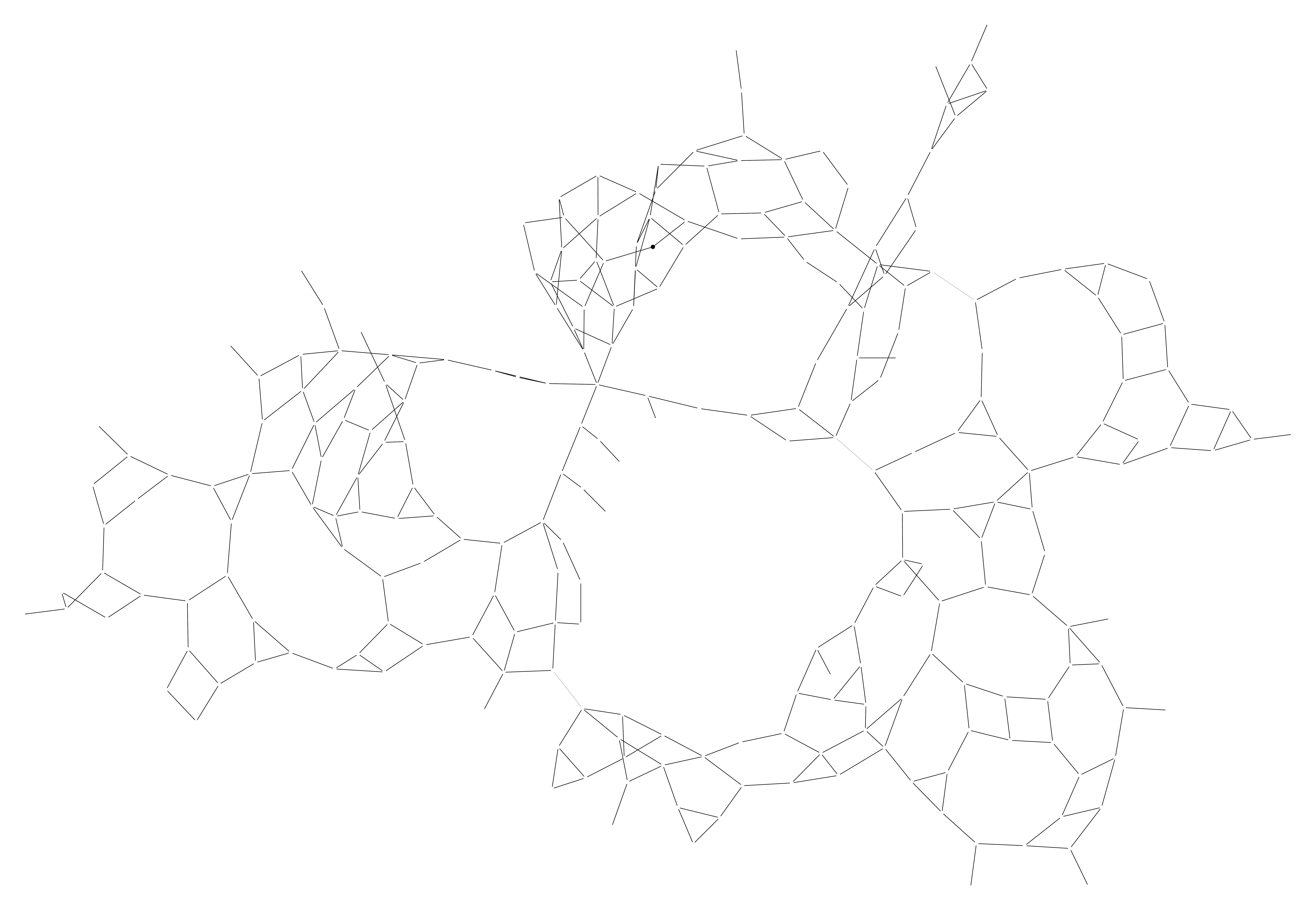}
    \caption{ob-egl-g2-E}
  \end{subfigure}
  \hfill
  \begin{subfigure}[b]{0.2\textwidth}
    \includegraphics[width=\textwidth]{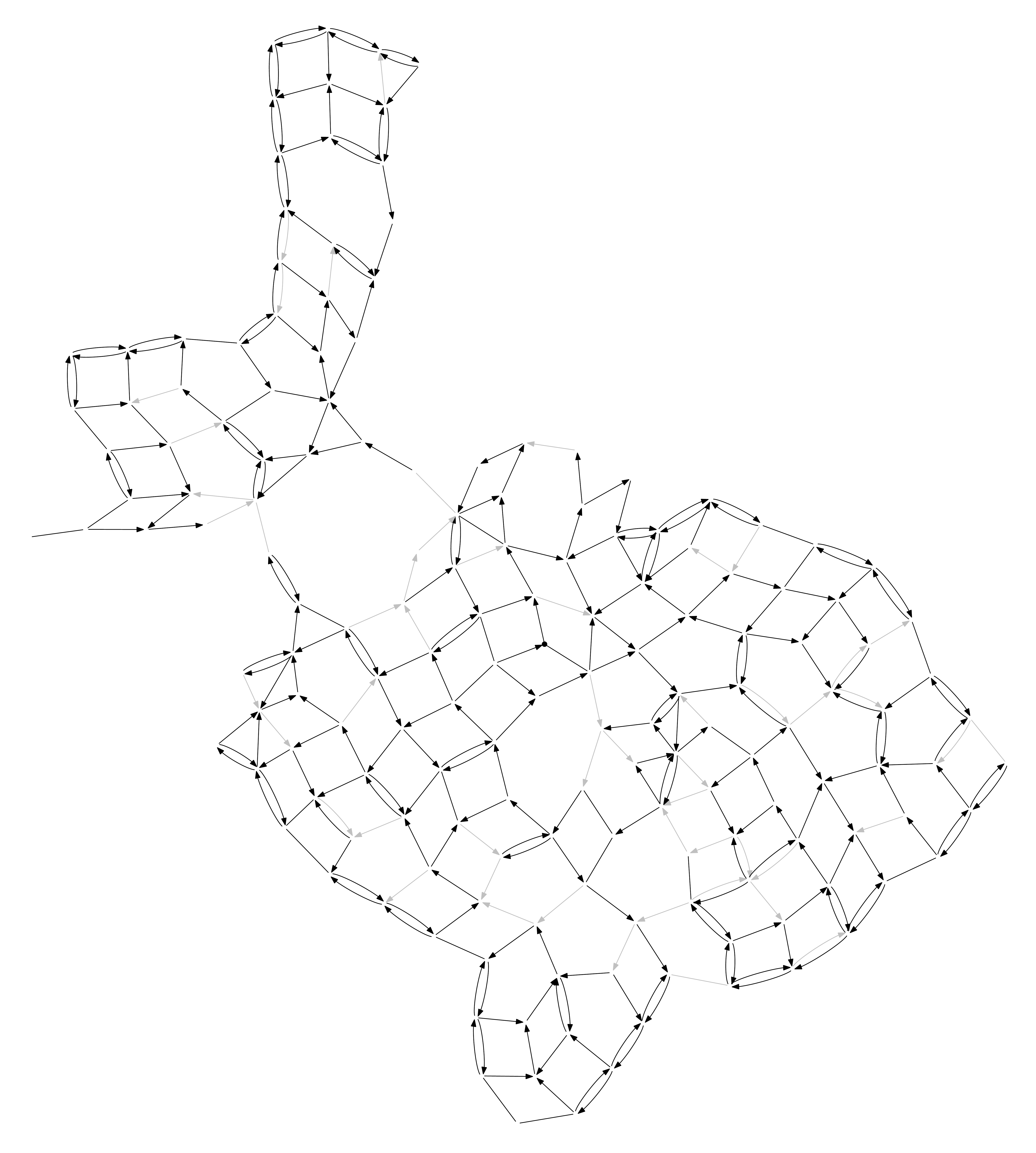}
    \caption{ob-lpr-b-03}
  \end{subfigure}
  \hfill
  \begin{subfigure}[b]{0.3\textwidth}
    \includegraphics[width=\textwidth]{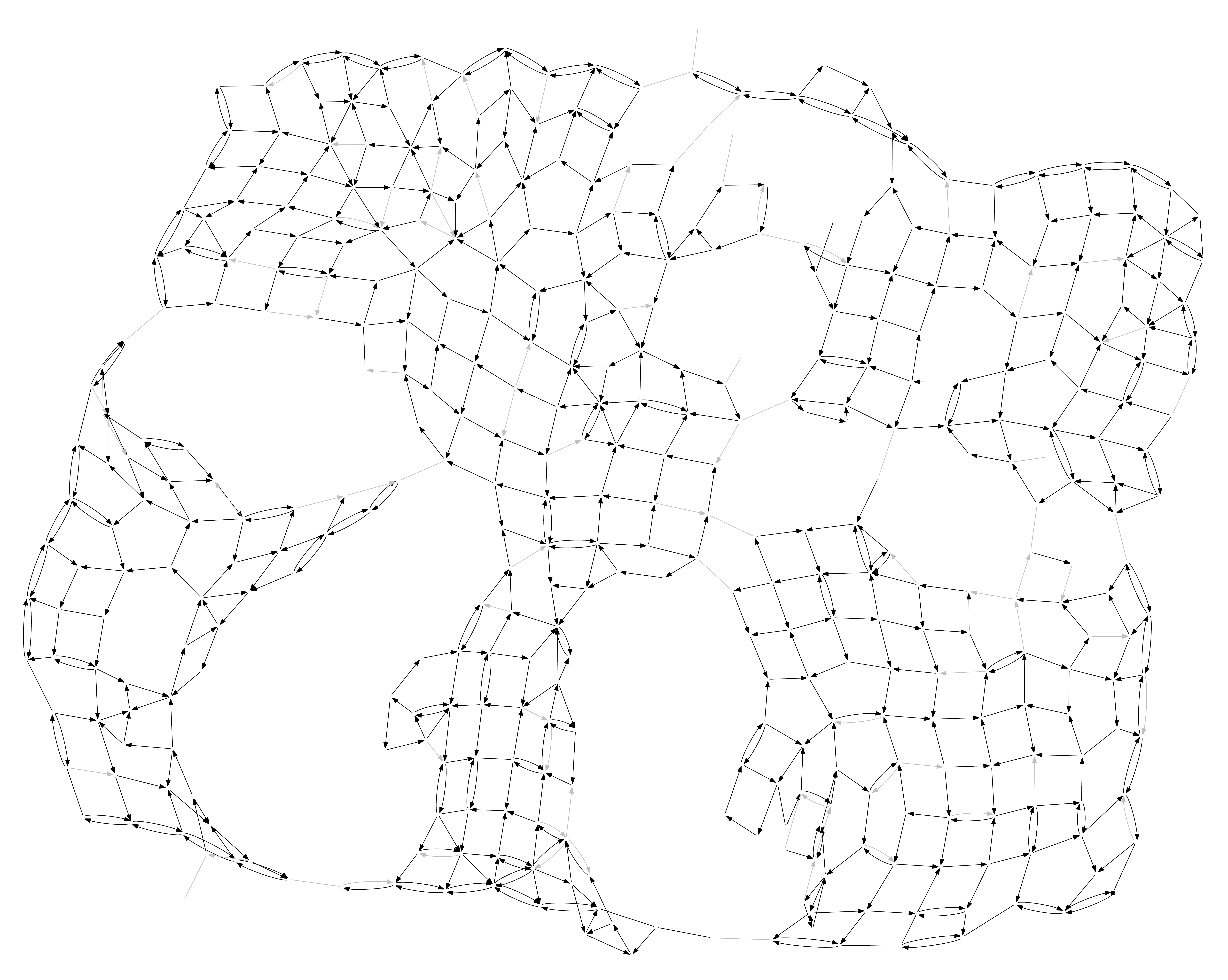}
    \caption{ob2-lpr-b-05}
  \end{subfigure}
  \caption{Three instances from the \texttt{Ob} benchmark set.}
  \label{fig:obinst}
\end{figure*}

\subsection{The Ob benchmark set}
\label{sec:ob}
Given our theoretical work in \cref{sec:rpp,sec:carp},
the solution quality achievable in polynomial time
appears to mainly depend on the number~\(C\) on connected components
in the graph induced by the positive-demand arcs.
However, we noticed that widely used benchmark instances for
variants of CARP have~\(C=1\).
In order to motivate a more representative evaluation
of the quality
of polynomial-time heuristics for variants of CARP,
we provide the \texttt{Ob} set of instances derived from the
\texttt{lpr} and \texttt{egl-large} instances with
\(C\)~from 2 to 5.  The approach can be easily used to create more
components.

The \texttt{Ob} instances%
\footnote{Available at \url{http://gitlab.com/rvb/mwcarp-ob}
and named after the river Ob, 
which bisects the city Novosibirsk.}
simulate cities that are divided by a river
that can be crossed via a few bridges without demand.
The underlying assumption is that, for example,
household waste does not have to be collected from bridges.
We generated the instances as follows.

As a base, we took sufficiently large instances from the \texttt{lpr} and \texttt{egl-large} sets
(it made little sense to split
the small \texttt{mval} or \texttt{lpr} instances
into several components).
In each instance,
we chose one or two random edges or arcs as ``bridges''.
Let \(B\)~be the set of their end points.
We then grouped all vertices of the graph into clusters:
For each \(v\in B\), there is one cluster containing all vertices
that are closer to~\(v\) than to all other vertices of~\(B\).
Finally,
we deleted all but a few edges between the clusters,
so that usually two or three edges
remain between each pair of clusters.
The demand of the edges remaining between clusters is set to zero,
they are our ``bridges'' between the river banks.
The intuition is that, if one of our initially chosen edges or arcs~\((u,v)\) was a bridge
across a relatively straight river,
then indeed every point on \(u\)'s side of the river would be closer
to~\(u\) than to~\(v\).
We discarded and regenerated instances
that were not strongly connected or had
river sides of highly imbalanced size
(three times below the average component size).
\cref{fig:obinst} shows three of the resulting instances.

Note that this approach can yield instances
where \(C\)~exceeds the number of clusters
since deleting edges between the clusters may create
more connected components in the graph
induced by the positive-demand arcs.
The approach straightforwardly applies to generating instances
with even larger~\(C\):
One simply chooses more initial ``bridges''.

As a starting point,
\cref{tab:obres} shows the number~\(C\),
a lower bound (LB) computed using an ILP relaxation of \citet{GMP10},
and the best upper bound obtained using our approximation algorithm
for each of the \texttt{Ob} instances
using any of the service direction heuristics in \cref{sec:serdir}.
The ``ob-'' instances were generated by choosing one initial bridge,
the ``ob2-'' instances were generated by choosing two initial bridges.

\begin{table*}[tb]
  \centering
  \caption{First upper and lower bounds for the \texttt{Ob} instances described in \cref{sec:ob}.}
  \label{tab:obres}
  \begin{tabular}{l|rrr||l|rrr}
\toprule
    Instance & $C$ & LB          & UB          & Instance   & $C$ & LB  & UB  \\
\midrule
ob-egl-g1-A & 2 & 817\,223    & 1\,152\,093 & ob2-egl-g1-A & 4 & 736\,899   & 1\,073\,386\\
ob-egl-g1-B & 2 & 1\,180\,105 & 1\,627\,305 & ob2-egl-g1-B & 5 & 840\,773   & 1\,221\,424\\
ob-egl-g1-C & 2 & 1\,018\,890 & 1\,405\,024 & ob2-egl-g1-C & 5 & 992\,974   & 1\,405\,836\\
ob-egl-g1-D & 3 & 1\,354\,671 & 1\,810\,306 & ob2-egl-g1-D & 4 & 1\,056\,593& 1\,491\,387\\
ob-egl-g1-E & 3 & 1\,486\,033 & 1\,955\,945 & ob2-egl-g1-E & 4 & 1\,175\,241& 1\,609\,377\\
ob-egl-g2-A & 2 & 922\,853    & 1\,286\,986 & ob2-egl-g2-A & 4 & 854\,823   & 1\,202\,379\\
ob-egl-g2-B & 2 & 1\,015\,013 & 1\,388\,809 & ob2-egl-g2-B & 4 & 906\,415   & 1\,259\,017\\
ob-egl-g2-C & 2 & 1\,308\,463 & 1\,701\,004 & ob2-egl-g2-C & 4 & 1\,154\,372& 1\,574\,762\\
ob-egl-g2-D & 2 & 1\,315\,717 & 1\,720\,548 & ob2-egl-g2-D & 4 & 1\,361\,397& 1\,782\,335\\
ob-egl-g2-E & 2 & 1\,677\,109 & 2\,139\,982 & ob2-egl-g2-E & 4 & 1\,295\,704& 1\,747\,883\\
\midrule
ob-lpr-a-03 & 3 & 71\,179     & 73\,055     & ob2-lpr-a-03 & 5 & 67\,219    & 69\,307  \\
ob-lpr-a-04 & 2 & 119\,759    & 123\,838    & ob2-lpr-a-04 & 4 & 115\,110   & 119\,550 \\
ob-lpr-a-05 & 2 & 195\,518    & 203\,832    & ob2-lpr-a-05 & 5 & 189\,968   & 197\,748 \\
ob-lpr-b-03 & 2 & 73\,670     & 75\,052     & ob2-lpr-b-03 & 5 & 67\,924    & 69\,518  \\
ob-lpr-b-04 & 2 & 122\,079    & 127\,020    & ob2-lpr-b-04 & 4 & 112\,104   & 116\,696 \\
ob-lpr-b-05 & 2 & 204\,389    & 213\,593    & ob2-lpr-b-05 & 5 & 191\,138   & 197\,878 \\
ob-lpr-c-03 & 2 & 105\,897    & 109\,913    & ob2-lpr-c-03 & 4 & 98\,244    & 102\,270 \\
ob-lpr-c-04 & 2 & 161\,856    & 167\,336    & ob2-lpr-c-04 & 4 & 155\,894   & 161\,615 \\
ob-lpr-c-05 & 2 & 250\,636    & 258\,396    & ob2-lpr-c-05 & 4 & 238\,299   & 246\,368 \\
\bottomrule
  \end{tabular}
\end{table*}

\section{Conclusion}
\label{sec:conc}
Since our algorithm outperforms many other polynomial-time heuristics,
it is useful for computing good solutions in instances
that are still too large to be attacked
by exact, local search, or genetic algorithms. Moreover, it might be useful to use our solution as initial solution for local search algorithms. 

Our theoretical results show that
one should not evaluate polynomial-time heuristics
only on instances whose positive-demand arcs induce
a graph with only one connected component,
because the solution quality achievable
in polynomial time is largely determined
by this number of connected components.
Therefore, it would be interesting to see
how other polynomial-time heuristics,
which do not take into account
the number of connected components in the graph
induced by the positive-demand arcs,
compare to our algorithm
in instances where this number is larger than one.

Finally, we conclude with a theoretical question:
It is easy to show a 3-approximation
for the Mixed Chinese Postman problem
using the approach in \cref{sec:mwrpp},
yet \citet{RV99} showed a $3/2$-\apxn{}.
Can our
$(\alpha(\comps)+3)$-\apxn{} for \MWRPP{} in
\cref{ourthm}\eqref{ourthm2} be improved to an
$(\alpha(\comps)+3/2)$-\apxn{} analogously?

\paragraph{Acknowledgments.} 
This research was initiated during a research retreat
of the algorithms and complexity theory group of TU Berlin,
held in Rothenburg/Oberlausitz, Germany, in March~2015.
We thank Sepp Hartung, Iyad Kanj, and André Nichterlein
for fruitful discussions.

{ 
  \setlength{\bibsep}{0pt}
  \bibliographystyle{abbrvnat}
  \bibliography{dar-arxiv}

\newcommand{\noopsort}[1]{}
\begin{thebibliography}{41}
\providecommand{\natexlab}[1]{#1}
\providecommand{\url}[1]{\texttt{#1}}
\expandafter\ifx\csname urlstyle\endcsname\relax
  \providecommand{\doi}[1]{doi: #1}\else
  \providecommand{\doi}{doi: \begingroup \urlstyle{rm}\Url}\fi

\bibitem[Ahuja et~al.(1993)Ahuja, Magnanti, and Orlin]{AMO93}
R.~K. Ahuja, T.~L. Magnanti, and J.~B. Orlin.
\newblock \emph{Network Flows---Theory, Algorithms and Applications}.
\newblock Prentice Hall, 1993.

\bibitem[Asadpour et~al.(2010)Asadpour, Goemans, M\k{a}dry, Gharan, and
  Saberi]{AGMOS10}
A.~Asadpour, M.~X. Goemans, A.~M\k{a}dry, S.~O. Gharan, and A.~Saberi.
\newblock An {$O(\log n/\allowbreak\log\log n)$}-ap\-prox\-i\-ma\-tion
  algorithm for the asymmetric traveling salesman problem.
\newblock In \emph{Proceedings of the 21st Annual ACM-SIAM Symposium on
  Discrete Algorithms (SODA'10)}, pages 379--389. Society for Industrial and
  Applied Mathematics, 2010.
\newblock \doi{10.1137/1.9781611973075.32}.

\bibitem[Beasley(1983)]{Bea83}
J.~Beasley.
\newblock Route first-{Cluster} second methods for vehicle routing.
\newblock \emph{Omega}, 11\penalty0 (4):\penalty0 403--408, 1983.
\newblock \doi{10.1016/0305-0483(83)90033-6}.

\bibitem[Belenguer et~al.(2006)Belenguer, Benavent, Lacomme, and Prins]{BBLP06}
J.-M. Belenguer, E.~Benavent, P.~Lacomme, and C.~Prins.
\newblock Lower and upper bounds for the mixed capacitated arc routing problem.
\newblock \emph{Computers {\&} Operations Research}, 33\penalty0 (12):\penalty0
  3363--3383, 2006.
\newblock \doi{10.1016/j.cor.2005.02.009}.

\bibitem[Bellman(1962)]{Bell62}
R.~Bellman.
\newblock Dynamic programming treatment of the {Travelling Salesman Problem}.
\newblock \emph{Journal of the ACM}, 9\penalty0 (1):\penalty0 61--63, 1962.
\newblock \doi{10.1145/321105.321111}.

\bibitem[{\noopsort{Bevern}van Bevern}
  et~al.(2014{\natexlab{a}}){\noopsort{Bevern}van Bevern}, Hartung,
  Nichterlein, and Sorge]{BHNS14}
R.~{\noopsort{Bevern}van Bevern}, S.~Hartung, A.~Nichterlein, and M.~Sorge.
\newblock Constant-factor approximations for capacitated arc routing without
  triangle inequality.
\newblock \emph{Operations Research Letters}, 42\penalty0 (4):\penalty0
  290--292, 2014{\natexlab{a}}.
\newblock \doi{10.1016/j.orl.2014.05.002}.

\bibitem[{\noopsort{Bevern}van Bevern}
  et~al.(2014{\natexlab{b}}){\noopsort{Bevern}van Bevern}, Niedermeier, Sorge,
  and Weller]{BNSW14}
R.~{\noopsort{Bevern}van Bevern}, R.~Niedermeier, M.~Sorge, and M.~Weller.
\newblock Complexity of arc routing problems.
\newblock In \emph{Arc Routing: Problems, Methods, and Applications}, MOS-SIAM
  Series on Optimization. SIAM, 2014{\natexlab{b}}.
\newblock \doi{10.1137/1.9781611973679.ch2}.

\bibitem[{\noopsort{Bevern}van Bevern} et~al.(2015){\noopsort{Bevern}van
  Bevern}, Komusiewicz, and Sorge]{BKS15}
R.~{\noopsort{Bevern}van Bevern}, C.~Komusiewicz, and M.~Sorge.
\newblock Approximation algorithms for mixed, windy, and capacitated arc
  routing problems.
\newblock In \emph{Proceedings of the 15th Workshop on Algorithmic Approaches
  for Transportation Modeling, Optimization, and Systems (ATMOS'15)}, volume~48
  of \emph{OpenAccess Series in Informatics (OASIcs)}, pages 130--143. Schloss
  Dagstuhl--Leibniz-Zentrum f{\"u}r Informatik, 2015.
\newblock \doi{10.4230/OASIcs.ATMOS.2015.130}.

\bibitem[Bode and Irnich(2015)]{BI15}
C.~Bode and S.~Irnich.
\newblock In-depth analysis of pricing problem relaxations for the capacitated
  arc-routing problem.
\newblock \emph{Transportation Science}, 49\penalty0 (2):\penalty0 369--383,
  2015.
\newblock \doi{10.1287/trsc.2013.0507}.

\bibitem[Brandão and Eglese(2008)]{BE08}
J.~Brandão and R.~Eglese.
\newblock A deterministic tabu search algorithm for the capacitated arc routing
  problem.
\newblock \emph{Computers {\&} Operations Research}, 35\penalty0 (4):\penalty0
  1112--1126, 2008.
\newblock \doi{10.1016/j.cor.2006.07.007}.

\bibitem[Christofides et~al.(1986)Christofides, Campos, Corber{\'a}n, and
  Mota]{CCCM86}
N.~Christofides, V.~Campos, {\'A}.~Corber{\'a}n, and E.~Mota.
\newblock An algorithm for the {Rural Postman} problem on a directed graph.
\newblock In \emph{Netflow at Pisa}, volume~26 of \emph{Mathematical
  Programming Studies}, pages 155--166. Springer, 1986.
\newblock \doi{10.1007/BFb0121091}.

\bibitem[Corberán and Laporte(2014)]{CL14}
{\' A}.~Corberán and G.~Laporte, editors.
\newblock \emph{Arc Routing: Problems, Methods, and Applications}.
\newblock SIAM, 2014.

\bibitem[Corberán et~al.(2000)Corberán, Martí, and Romero]{CMR00}
A.~Corberán, R.~Martí, and A.~Romero.
\newblock Heuristics for the mixed rural postman problem.
\newblock \emph{Computers {\&} Operations Research}, 27\penalty0 (2):\penalty0
  183--203, 2000.
\newblock \doi{10.1016/S0305-0548(99)00031-3}.

\bibitem[Cygan et~al.(2015)Cygan, Fomin, Kowalik, Lokshtanov, Marx, Pilipczuk,
  Pilipczuk, and Saurabh]{CFK+15}
M.~Cygan, F.~V. Fomin, L.~Kowalik, D.~Lokshtanov, D.~Marx, M.~Pilipczuk,
  M.~Pilipczuk, and S.~Saurabh.
\newblock \emph{Parameterized Algorithms}.
\newblock Springer, 2015.
\newblock \doi{10.1007/978-3-319-21275-3}.

\bibitem[Ding et~al.(2014)Ding, Li, and Lih]{DLL14}
H.~Ding, J.~Li, and K.-W. Lih.
\newblock Approximation algorithms for solving the constrained arc routing
  problem in mixed graphs.
\newblock \emph{European Journal of Operational Research}, 239\penalty0
  (1):\penalty0 80--88, 2014.
\newblock \doi{10.1016/j.ejor.2014.04.039}.

\bibitem[Dorn et~al.(2013)Dorn, Moser, Niedermeier, and Weller]{DMNW13}
F.~Dorn, H.~Moser, R.~Niedermeier, and M.~Weller.
\newblock Efficient algorithms for {Eulerian Extension} and {Rural Postman}.
\newblock \emph{SIAM Journal on Discrete Mathematics}, 27\penalty0
  (1):\penalty0 75--94, 2013.
\newblock \doi{10.1137/110834810}.

\bibitem[Downey and Fellows(2013)]{DF13}
R.~G. Downey and M.~R. Fellows.
\newblock \emph{Fundamentals of Parameterized Complexity}.
\newblock Springer, 2013.
\newblock \doi{10.1007/978-1-4471-5559-1}.

\bibitem[Edmonds(1975)]{Edm65}
J.~Edmonds.
\newblock The {C}hinese postman problem.
\newblock \emph{Operations Research}, pages B73--B77, 1975.
\newblock Supplement\,1.

\bibitem[Edmonds and Johnson(1973)]{EJ73}
J.~Edmonds and E.~L. Johnson.
\newblock Matching, {E}uler tours and the {C}hinese postman.
\newblock \emph{Mathematical Programming}, 5:\penalty0 88--124, 1973.
\newblock \doi{10.1007/BF01580113}.

\bibitem[Floyd(1962)]{Flo62}
R.~W. Floyd.
\newblock Algorithm 97: Shortest path.
\newblock \emph{Communications of the ACM}, 5\penalty0 (6):\penalty0 345, 1962.
\newblock \doi{10.1145/367766.368168}.

\bibitem[Frederickson(1977)]{Fre77}
G.~N. Frederickson.
\newblock \emph{Approximation {A}lgorithms for {NP}-hard {R}outing {P}roblems}.
\newblock PhD thesis, Faculty of the Graduate School of the University of
  Maryland, 1977.

\bibitem[Frederickson(1979)]{Fre79}
G.~N. Frederickson.
\newblock Approximation algorithms for some postman problems.
\newblock \emph{Journal of the ACM}, 26\penalty0 (3):\penalty0 538--554, 1979.
\newblock \doi{10.1145/322139.322150}.

\bibitem[Frederickson et~al.(1978)Frederickson, Hecht, and Kim]{FHK78}
G.~N. Frederickson, M.~S. Hecht, and C.~E. Kim.
\newblock Approximation algorithms for some routing problems.
\newblock \emph{SIAM Journal on Computing}, 7\penalty0 (2):\penalty0 178--193,
  1978.
\newblock \doi{10.1137/0207017}.

\bibitem[Frieze et~al.(1982)Frieze, Galbiati, and Maffioli]{FGM82}
A.~M. Frieze, G.~Galbiati, and F.~Maffioli.
\newblock On the worst-case performance of some algorithms for the asymmetric
  traveling salesman problem.
\newblock \emph{Networks}, 12\penalty0 (1):\penalty0 23--39, 1982.
\newblock \doi{10.1002/net.3230120103}.

\bibitem[Golden and Wong(1981)]{GW81}
B.~L. Golden and R.~T. Wong.
\newblock Capacitated arc routing problems.
\newblock \emph{Networks}, 11\penalty0 (3):\penalty0 305--315, 1981.
\newblock \doi{10.1002/net.3230110308}.

\bibitem[Golden et~al.(1983)Golden, Dearmon, and Baker]{GDB83}
B.~L. Golden, J.~S. Dearmon, and E.~K. Baker.
\newblock Computational experiments with algorithms for a class of routing
  problems.
\newblock \emph{Computers {\&} Operations Research}, 10\penalty0 (1):\penalty0
  47--59, 1983.
\newblock \doi{10.1016/0305-0548(83)90026-6}.

\bibitem[Gouveia et~al.(2010)Gouveia, Mourão, and Pinto]{GMP10}
L.~Gouveia, M.~C. Mourão, and L.~S. Pinto.
\newblock Lower bounds for the mixed capacitated arc routing problem.
\newblock \emph{Computers \& Operations Research}, 37\penalty0 (4):\penalty0
  692--699, 2010.
\newblock \doi{10.1016/j.cor.2009.06.018}.

\bibitem[Gutin et~al.(2016)Gutin, Wahlstr{\"o}m, and Yeo]{GWY16}
G.~Gutin, M.~Wahlstr{\"o}m, and A.~Yeo.
\newblock {Rural Postman} parameterized by the number of components of required
  edges.
\newblock \emph{Journal of Computer and System Sciences}, 2016.
\newblock \doi{10.1016/j.jcss.2016.06.001}.
\newblock In press.

\bibitem[Hall(1935)]{Hal35}
P.~Hall.
\newblock On representatives of subsets.
\newblock \emph{Journal of the London Mathematical Society}, 10:\penalty0
  26--30, 1935.
\newblock \doi{10.1112/jlms/s1-10.37.26}.

\bibitem[Held and Karp(1962)]{HelK62}
M.~Held and R.~M. Karp.
\newblock A dynamic programming approach to sequencing problems.
\newblock \emph{Journal of the Society for Industrial and Applied Mathematics},
  10\penalty0 (1):\penalty0 196--210, 1962.
\newblock \doi{10.1137/0110015}.

\bibitem[Jansen(1993)]{Jan93}
K.~Jansen.
\newblock Bounds for the general capacitated routing problem.
\newblock \emph{Networks}, 23\penalty0 (3):\penalty0 165--173, 1993.
\newblock \doi{10.1002/net.3230230304}.

\bibitem[Lenstra and {Rinnooy Kan}(1976)]{LR76}
J.~K. Lenstra and A.~H.~G. {Rinnooy Kan}.
\newblock On general routing problems.
\newblock \emph{Networks}, 6\penalty0 (3):\penalty0 273--280, 1976.
\newblock \doi{10.1002/net.3230060305}.

\bibitem[Marx(2008)]{Mar08}
D.~Marx.
\newblock Parameterized complexity and approximation algorithms.
\newblock \emph{The Computer Journal}, 51\penalty0 (1):\penalty0 60--78, 2008.
\newblock \doi{10.1093/comjnl/bxm048}.

\bibitem[Mourão and Amado(2005)]{MA05}
M.~C. Mourão and L.~Amado.
\newblock Heuristic method for a mixed capacitated arc routing problem: A
  refuse collection application.
\newblock \emph{European Journal of Operational Research}, 160\penalty0
  (1):\penalty0 139--153, 2005.
\newblock \doi{10.1016/j.ejor.2004.01.023}.

\bibitem[Orloff(1976)]{Orl76}
C.~S. Orloff.
\newblock On general routing problems: Comments.
\newblock \emph{Networks}, 6\penalty0 (3):\penalty0 281--284, 1976.
\newblock \doi{10.1002/net.3230060306}.

\bibitem[Pearn and Wu(1995)]{PW95}
W.~L. Pearn and T.~C. Wu.
\newblock Algorithms for the rural postman problem.
\newblock \emph{Computers {\&} Operations Research}, 22\penalty0 (8):\penalty0
  819--828, 1995.
\newblock \doi{10.1016/0305-0548(94)00070-O}.

\bibitem[Raghavachari and Veerasamy(1999)]{RV99}
B.~Raghavachari and J.~Veerasamy.
\newblock A 3/2-approximation algorithm for the {Mixed Postman Problem}.
\newblock \emph{SIAM Journal on Discrete Mathematics}, 12\penalty0
  (4):\penalty0 425--433, 1999.
\newblock \doi{10.1137/S0895480197331454}.

\bibitem[Sorge et~al.(2011)Sorge, van Bevern, Niedermeier, and Weller]{SBNW11}
M.~Sorge, R.~van Bevern, R.~Niedermeier, and M.~Weller.
\newblock From few components to an {Eulerian} graph by adding arcs.
\newblock In \emph{Proceedings of the 37th International Workshop on
  Graph-Theoretic Concepts in Computer Science (WG'11)}, pages 307--318.
  Springer, 2011.
\newblock \doi{10.1007/978-3-642-25870-1_28}.

\bibitem[Sorge et~al.(2012)Sorge, van Bevern, Niedermeier, and Weller]{SBNW12}
M.~Sorge, R.~van Bevern, R.~Niedermeier, and M.~Weller.
\newblock A new view on {Rural Postman} based on {Eulerian Extension} and
  {Matching}.
\newblock \emph{Journal of Discrete Algorithms}, 16:\penalty0 12--33, 2012.
\newblock \doi{10.1016/j.jda.2012.04.007}.

\bibitem[Ulusoy(1985)]{Ulu85}
G.~Ulusoy.
\newblock The fleet size and mix problem for capacitated arc routing.
\newblock \emph{European Journal of Operational Research}, 22\penalty0
  (3):\penalty0 329--337, 1985.
\newblock \doi{10.1016/0377-2217(85)90252-8}.

\bibitem[Wøhlk(2008)]{Woe08}
S.~Wøhlk.
\newblock An approximation algorithm for the {Capacitated Arc Routing Problem}.
\newblock \emph{The Open Operational Research Journal}, 2:\penalty0 8--12,
  2008.
\newblock \doi{10.2174/1874243200802010008}.

\end{thebibliography}
}
\end{document}